\documentclass{article}
\usepackage[numbers,sort]{natbib}
\usepackage{amsthm}
\usepackage{authblk}

\makeatletter
\newtheorem*{rep@theorem}{\rep@title}
\newcommand{\newreptheorem}[2]{%
\newenvironment{rep#1}[1]{%
 \def\rep@title{#2 \ref{##1}}%
 \begin{rep@theorem}}%
 {\end{rep@theorem}}}
\makeatother

\theoremstyle{definition}
\newtheorem{theorem}{Theorem}[section]
\newreptheorem{theorem}{Theorem}
\newtheorem{lemma}{Lemma}[section]
\newtheorem{corollary}{Corollary}[section]
\newtheorem{proposition}{Proposition}[section]%
\newreptheorem{proposition}{Proposition}
\newtheorem{fact}{Fact}[section]%

\newtheorem{definition}{Definition}[section]
\newtheorem{rrule}{Reduction Rule}[section]%
\newtheorem{srule}{Switching Rule}[section]%
\newtheorem{brule}{Branching Rule}[section]%

\usepackage{color}
\usepackage[utf8]{inputenc}
\usepackage[T1]{fontenc}
\usepackage{lmodern}
\usepackage{mathtools}
\usepackage{amsmath,amssymb}
\usepackage{enumerate}
\usepackage{graphicx}
\usepackage{booktabs}
\usepackage{xcolor}
\usepackage[textsize=footnotesize]{todonotes}
\usepackage{verbatim}
\usepackage{paralist}
\usepackage{multicol}
\newcommand{\cref}[1]{\todo[inline]{Error!}}
\newcommand{\brref}[1]{Branching Rule~\ref{#1}}
\newcommand{\brrefm}[1]{Branching Rules~\ref{#1}}
\newcommand{\rrref}[1]{Reduction Rule~\ref{#1}}
\newcommand{\rrrefm}[1]{Reduction Rules~\ref{#1}}
\newcommand{\srref}[1]{Switching Rule~\ref{#1}}
\newcommand{\srrefm}[1]{Switching Rules~\ref{#1}}

\newlength{\capitalheight}
\newcommand{\Oh}{\mathcal{O}}

\renewcommand{\P}{\text{\normalfont P}}
\newcommand{\NP}{\text{\normalfont NP}}
\newcommand{\FPT}{\text{\normalfont FPT}}
\newcommand{\XP}{\text{\normalfont XP}}
\newcommand{\W}[1]{\text{\normalfont W[#1]}}

\newcommand{\co}[1]{\overline{#1}}

\newcommand{\sA}{A_{*}^{\C}}
\newcommand{\sB}{B_{*}^{\C}}
\newcommand{\ssA}{A^{\C}}
\newcommand{\ssB}{B^{\C}}
\newcommand{\pA}{A_{P}^{\C}}
\newcommand{\pB}{B_{P}^{\C}}

\newcommand{\C}{{\cal C}}
\newcommand{\rC}{R^{\C}}
\DeclareMathOperator{\switch}{switch}
\newcommand{\abpartition}{$(\Pi_A,\Pi_B)$-partition}

\def\etal{{\em et al.}}

\newcommand{\smartqed}{}

\title{Parameterized Algorithms for Recognizing Monopolar and 2-Subcolorable Graphs\footnote{A preliminary version of this paper appeared in SWAT 2016, volume 53 of LIPIcs, pages~14:1--14:14. Iyad Kanj and Manuel Sorge gratefully acknowledge the support by Deutsche Forschungsgemeinschaft (DFG), project DAPA, NI 369/12. Christian Komusiewicz gratefully acknowledges the support by Deutsche Forschungsgemeinschaft (DFG), project MAGZ, KO 3669/4-1.}}

\author[1]{Iyad Kanj}

\author[2]{Christian Komusiewicz}

\author[3]{Manuel Sorge}

\author[4]{\mbox{Erik Jan van Leeuwen}}

\affil[1]{School of Computing, DePaul University, Chicago, USA, \texttt{ikanj@cs.depaul.edu}}
\affil[2]{Institut f{\"u}r Informatik, Friedrich-Schiller-Universit{\"a}t Jena, Germany, \texttt{christian.komusiewicz@uni-jena.de}}
\affil[3]{Institut f\"{u}r Softwaretechnik und Theoretische Informatik, TU Berlin, Germany, \texttt{manuel.sorge@tu-berlin.de}}
\affil[4]{Max-Planck-Institut f\"ur Informatik, Saarland Informatics Campus, Saarbr\"ucken, Germany, \texttt{erikjan@mpi-inf.mpg.de}}

\newcommand{\thmtxtmainmono}{In $\Oh(2^{k} \cdot k^3 \cdot (n+m))$ time, we can decide whether $G$ admits a monopolar partition $(A,B)$ such that $G[A]$ is a cluster graph with at most $k$~clusters.}

\newcommand{\thmtxtexclusiverec}{If~$\Pi_{A}$ and~$\Pi_B$ are hereditary and mutually~$d$-exclusive, and
  membership of~$\Pi_A$ and~$\Pi_B$ can be decided in polynomial time, then
  \textsc{$(\Pi_{A},\Pi_{B})$-Recognition} can be solved in~$n^{2d+\Oh(1)}$ time.}

\newcommand{\thmtxtmainsub}{In $\Oh(k^{2k+1}\cdot nm)$ time, we can decide whether $G$ admits a 2-subcoloring $(A,B)$ such that $G[A]$ is a cluster graph with at most $k$~clusters.}

\newcommand{\thmtxtmainsubtwo}{In $\Oh(4^{k}\cdot k^2 \cdot n^2)$ time, we can decide whether $G$ admits a 2\nobreakdash-subcoloring $(A,B)$ such that $G[A]$ and $G[B]$ are cluster graphs with at most $k$ clusters in total. }

\newcommand{\prptxtmainsize}{Let~$\Pi_A$ and~$\Pi_B$ be two hereditary graph properties such that membership of $\Pi_{A}$ can be decided in polynomial time and~$\Pi_B$ can be characterized by a finite set of forbidden induced subgraphs. Then we can decide in~$2^{\Oh(k)}n^{\Oh(1)}$ time whether~$V$ can be partitioned into sets $A$ and~$B$ such that~$G[A]\in \Pi_A$,~$G[B]\in \Pi_B$, and~$|A|\le k$.}

\begin{document}

\maketitle

\begin{abstract}
  \noindent A graph $G$ is a $(\Pi_A,\Pi_B)$-graph if $V(G)$ can be bipartitioned into $A$ and $B$
  such that $G[A]$ satisfies property $\Pi_A$ and $G[B]$ satisfies property $\Pi_B$. The
  \textsc{$(\Pi_{A},\Pi_{B})$-Recognition} problem is to recognize whether a given graph is a
  $(\Pi_A,\Pi_B)$-graph. There are many \textsc{$(\Pi_{A},\Pi_{B})$-Recognition} problems,
  including the recognition problems for bipartite, split, and unipolar graphs. We present
  efficient algorithms for many cases of \textsc{$(\Pi_A,\Pi_B)$-Recognition} based on a technique
  which we dub inductive recognition. In particular, we give fixed-parameter algorithms
  for two NP-hard \textsc{$(\Pi_{A},\Pi_{B})$-Recognition} problems, \textsc{Monopolar Recognition} and
  \textsc{2-Subcoloring}, parameterized by the number of maximal cliques in $G[A]$. We complement our algorithmic results with several hardness results for
  \textsc{$(\Pi_{A},\Pi_{B})$-Recognition}.
\end{abstract}

\section{Introduction} \label{sec:intro}

A \emph{$(\Pi_A,\Pi_B)$-graph}, for graph properties $\Pi_{A},\Pi_{B}$, is a graph $G=(V,E)$ for which $V$ admits a partition into two sets $A, B$ such that $G[A]$ satisfies $\Pi_A$ and $G[B]$ satisfies $\Pi_B$. There is an abundance of $(\Pi_{A},\Pi_{B})$-graph classes, and important ones include \emph{bipartite graphs} (which admit a partition into two independent sets), \emph{split graphs} (which admit a bipartition into a clique and an independent set), and \emph{unipolar graphs} (which admit a bipartition into a clique and a cluster graph). Here a \emph{cluster graph} is a disjoint union of cliques. An example for each of these classes is shown in Figure~\ref{fig:abgraphs}.
\begin{figure}[t]
  \centering
  \begin{tikzpicture}%

    \tikzstyle{avertex} = [color=black,fill=black,circle,inner sep=0pt,minimum size=8pt]
    \tikzstyle{bvertex} = [color=black,draw,circle,semithick,inner sep=0pt,minimum size=8pt]

    \begin{scope}[shift={(0,0)}]
        \node[avertex] (a1) at (1,0) {};
        \node[avertex] (a2) at (2,0) {};
        \node[avertex] (a3) at (3,0) {};
        \node[bvertex] (b1) at (1,1) {};
        \node[bvertex] (b2) at (2,1) {};
        \node[bvertex] (b3) at (3,1) {};
        \draw[semithick] (a1)--(b1);
        \draw[semithick] (a1)--(b2);
        \draw[semithick] (a3)--(b3);
        \draw[semithick] (a2)--(b2);
        \draw[semithick] (a3)--(b2);
  \end{scope}

  \begin{scope}[shift={(5,0)}]
        \node[avertex] (a1) at (1.2,0) {};
        \node[avertex] (a2) at (2,-1) {};
        \node[avertex] (a3) at (2.8,0) {};
        \node[bvertex] (b1) at (1,1) {};
        \node[bvertex] (b2) at (2,1) {};
        \node[bvertex] (b3) at (3,1) {};
        \draw[semithick] (a1)--(a2)--(a3)--(a1);
        \draw[semithick] (a1)--(b1);
        \draw[semithick] (a1)--(b2);
        \draw[semithick] (a3)--(b3);
        \draw[semithick] (a2)--(b2);
        \draw[semithick] (a3)--(b2);
  \end{scope}

  \begin{scope}[shift={(10,0)}]
        \node[avertex] (a1) at (1.2,0) {};
        \node[avertex] (a2) at (2,-1) {};
        \node[avertex] (a3) at (2.8,0) {};
        \node[bvertex] (b11) at (0.5,1.8) {};
        \node[bvertex] (b12) at (1,1) {};
        \node[bvertex] (b13) at (0,1) {};
        \node[bvertex] (b21) at (2,1) {};
        \node[bvertex] (b22) at (2,1.8) {};
        \node[bvertex] (b3) at (3,1) {};
        \draw[semithick] (a1)--(a2)--(a3)--(a1);
        \draw[semithick] (a1)--(b12);
        \draw[semithick] (a1)--(b13);
        \draw[semithick] (a3)--(b12);
        \draw[semithick] (a1)--(b21);
        \draw[semithick] (a3)--(b3);
        \draw[semithick] (a2)--(b21);
        \draw[semithick] (a3)--(b21);
        \draw[semithick] (b11)--(b12)--(b13)--(b11);
        \draw[semithick] (b22)--(b21);
  \end{scope}

  \end{tikzpicture}
  \caption{Three examples of $(\Pi_{A},\Pi_{B})$-graphs, where the coloring gives a~\abpartition{}. The vertices of~$A$ are black and the vertices of~$B$ are white. Left: in bipartite graphs,~$A$ and~$B$ are independent sets. Center: in split graphs,~$A$ is a clique and~$B$ is an independent set. Right: in unipolar graphs,~$A$ is a clique and~$B$ induces a cluster graph.}
\label{fig:abgraphs}
\end{figure}
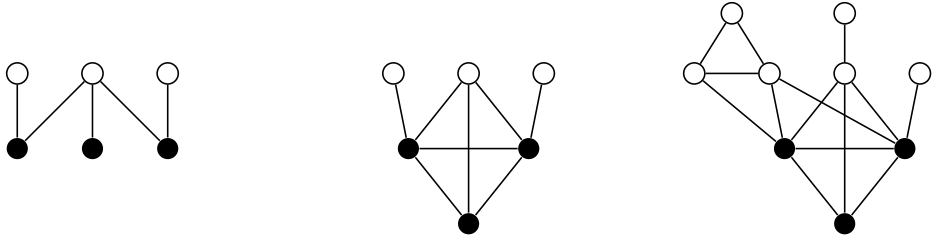

The problem of recognizing whether a given graph belongs to a particular class of $(\Pi_{A},\Pi_{B})$-graphs is called \textsc{$(\Pi_{A},\Pi_{B})$-Recognition}, and is known as a \emph{vertex-partition} problem. Recognition problems for $(\Pi_A,\Pi_B)$-graphs are often NP-hard~\cite{Ach97,Far04,KS97}, but bipartite, split, and unipolar graphs can all be recognized in polynomial time~\cite{papa,HS81,MY15,unipolar,unipolar0}. With the aim of generalizing these polynomial-time algorithms, we study the complexity of recognizing certain classes of $(\Pi_A,\Pi_B)$-graphs, focusing on two particular classes that generalize split and unipolar graphs, respectively.
To achieve our goals, we formalize a technique, which we dub inductive recognition, that can be viewed as an adaptation of iterative compression to recognition problems. We believe that the formalization of this technique will be useful in general for designing algorithms for recognition problems.

\paragraph*{Inductive Recognition}
The \emph{inductive recognition} technique, described formally in Section~\ref{sec:inducregoc}, can be applied to solve the \textsc{$(\Pi_{A},\Pi_{B})$-Recognition} problem for certain hereditary $(\Pi_A,\Pi_B)$-graph classes. Intuitively, the technique works as follows. Suppose that we are given a graph $G=(V,E)$ and we have to decide its membership of the $(\Pi_A,\Pi_B)$-graph class.
We proceed in iterations and fix an arbitrary ordering of the vertices; in the following, let~$n:=|V|$ and~$m:=|E|$. We start with the empty graph $G_0$, which trivially belongs to the hereditary $(\Pi_A, \Pi_B)$-graph class. In iteration $i$, for $i =1, \ldots, n$, we recognize whether the subgraph $G_i$ of $G$ induced by the first $i$~vertices of $V$ still belongs to the graph class, assuming that $G_{i-1}$ belongs to the graph class.

Inductive recognition is essentially a variant of the iterative compression technique~\cite{RSV04}, tailored to recognition
problems.  The crucial difference, however, is that in iterative compression we can always add the $i$th vertex $v_i$ to the solution from the previous iteration to obtain a new solution (which we compress if it is too large). However, in the recognition problems under consideration, we cannot simply add~$v_i$ to one part of a bipartition $(A, B)$ of $G_{i-1}$, where $G_{i-1}$ is member of the graph class, and witness that $G_i$ is still a member of the graph class: Adding $v_i$ to $A$ may violate property $\Pi_{A}$ {\em and} adding $v_i$ to $B$ may violate property $\Pi_{B}$. An example for split graph recognition is presented in Figure~\ref{fig:ind-step}. Here, we cannot add~$v_i$ to $A$ or $B$ to obtain a valid bipartition for~$G_i$, even if $G_{i-1}$ is a split graph with clique~$A$ and independent set~$B$. Therefore, we cannot perform a `compression step' as in iterative compression. Instead, we must attempt to add $v_i$ to each of $A$ and $B$, and then attempt to `repair' the resulting partition in each of the two cases, by rearranging vertices, into a solution for $G_i$ (if a solution exists). This idea is formalized in the inductive recognition framework in Section~\ref{sec:inducregoc}.
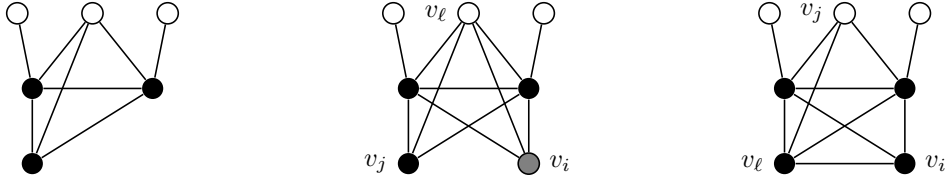
\begin{figure}[t]
  \centering
  \begin{tikzpicture}%

    \tikzstyle{avertex} = [color=black,fill=black,circle,inner sep=0pt,minimum size=8pt]
    \tikzstyle{bvertex} = [color=black,draw,circle,semithick,inner sep=0pt,minimum size=8pt]
    \tikzstyle{nvertex} = [color=black,fill=gray,draw,circle,semithick,inner sep=0pt,minimum size=8pt]

    \begin{scope}[shift={(0,0)}]
        \node[avertex] (a1) at (1.2,0) {};
        \node[avertex] (a2) at (1.2,-1) {};
        \node[avertex] (a3) at (2.8,0) {};
        \node[bvertex] (b1) at (1,1) {};
        \node[bvertex] (b2) at (2,1) {};
        \node[bvertex] (b3) at (3,1) {};
        \draw[semithick] (a1)--(a2)--(a3)--(a1);
        \draw[semithick] (a1)--(b1);
        \draw[semithick] (a1)--(b2);
        \draw[semithick] (a3)--(b3);
        \draw[semithick] (a2)--(b2);
        \draw[semithick] (a3)--(b2);
  \end{scope}

    \begin{scope}[shift={(5,0)}]
        \node[avertex] (a1) at (1.2,0) {};
        \node[avertex, label =left:$v_j$] (a2) at (1.2,-1) {};
        \node[avertex] (a3) at (2.8,0) {};
        \node[bvertex] (b1) at (1,1) {};
        \node[bvertex,label =left:$v_\ell$] (b2) at (2,1) {};
        \node[bvertex] (b3) at (3,1) {};
        \node[nvertex, label =right:$v_i$] (v) at (2.8,-1) {};
        \draw[semithick] (a1)--(a2)--(a3)--(a1);
        \draw[semithick] (a1)--(b1);
        \draw[semithick] (a1)--(b2);
        \draw[semithick] (a3)--(b3);
        \draw[semithick] (a2)--(b2);
        \draw[semithick] (a3)--(b2);
        \draw[semithick] (a1)--(v)--(b2);
        \draw[semithick] (a3)--(v);
  \end{scope}

    \begin{scope}[shift={(10,0)}]
        \node[avertex] (a1) at (1.2,0) {};
        \node[bvertex, label =left:$v_j$] (a2) at (2,1) {};
        \node[avertex] (a3) at (2.8,0) {};
        \node[bvertex] (b1) at (1,1) {};
        \node[avertex,label =left:$v_\ell$] (b2) at (1.2,-1) {};
        \node[bvertex] (b3) at (3,1) {};
        \node[avertex, label =right:$v_i$] (v) at (2.8,-1) {};
        \draw[semithick] (a1)--(a2)--(a3)--(a1);
        \draw[semithick] (a1)--(b1);
        \draw[semithick] (a1)--(b2);
        \draw[semithick] (a3)--(b3);
        \draw[semithick] (a2)--(b2);
        \draw[semithick] (a3)--(b2);
        \draw[semithick] (a1)--(v)--(b2);
        \draw[semithick] (a3)--(v);
  \end{scope}

  \end{tikzpicture}
  \caption{An example of the inductive step in inductive recognition. Left: a split graph $G_{i-1}$ with a given partition into a clique~$A$ and an independent set~$B$. Center: the~partition for~$G_{i-1}$ cannot be directly extended to a partition for~$G_i$ since the vertex~$v_i$ has a nonneighbor~$v_j$ in~$A$ and a neighbor~$v_\ell$ in~$B$. Right: after deciding to put~$v_i$ in the clique~$A$, we can repair the partition by moving~$v_j$ to the independent set~$B$ and~$v_\ell$ to the clique~$A$.}
\label{fig:ind-step}
\end{figure}

\paragraph*{Monopolar Graphs and Mutually Exclusive Graph Properties}
The first \textsc{$(\Pi_{A},\Pi_{B})$-Recognition} problem that we consider is the problem of recognizing monopolar graphs, which are a superset of
split graphs. A \emph{monopolar graph} is a graph whose vertex set
admits a bipartition into a cluster graph and an independent set; an example is shown in Figure~\ref{fig:mono-2sub}. Monopolar
graphs have applications in the analysis of protein-interaction networks~\cite{BHK14}. The
recognition problem of monopolar graphs can be formulated as follows:
\begin{quote}
  \textsc{Monopolar Recognition}\\
  \textbf{Input:} A graph $G=(V,E)$.\\
  \textbf{Question:} Does~$G$ have a \emph{monopolar partition}~$(A,B)$, that is, can $V$ be partitioned into sets $A$ and~$B$ such that $G[A]$~is a cluster graph and~$G[B]$ is an edgeless graph?
\end{quote}
\begin{figure}[t]
  \centering
  \begin{tikzpicture}%

    \tikzstyle{avertex} = [color=black,fill=black,circle,inner sep=0pt,minimum size=8pt]
    \tikzstyle{bvertex} = [color=black,draw,circle,semithick,inner sep=0pt,minimum size=8pt]

  \begin{scope}[shift={(0,0)}]
        \node[avertex] (a1) at (1.2,0) {};
        \node[avertex] (a2) at (2,-1) {};
        \node[avertex] (a3) at (2.8,0) {};
        \node[avertex] (a4) at (4.2,0) {};
        \node[avertex] (a5) at (5,-1) {};
        \node[avertex] (a6) at (5.8,0) {};

        \node[bvertex] (b1) at (0.8,1) {};
        \node[bvertex] (b2) at (2,1) {};
        \node[bvertex] (b3) at (3.45,1) {};
        \node[bvertex] (b4) at (5,1) {};
        \draw[semithick] (a1)--(a2)--(a3)--(a1);
        \draw[semithick] (a4)--(a5)--(a6)--(a4);
        \draw[semithick] (a1)--(b1);
        \draw[semithick] (a3)--(b3)--(a4);
        \draw[semithick] (a2)--(b2)--(a1);
        \draw[semithick] (a4)--(b4)--(a6);
  \end{scope}

  \begin{scope}[shift={(7,0)}]
        \node[avertex] (a1) at (1.2,0) {};
        \node[avertex] (a2) at (2,-1) {};
        \node[avertex] (a3) at (2.8,0) {};
        \node[avertex] (a4) at (4.2,0) {};
        \node[avertex] (a5) at (5,-1) {};
        \node[avertex] (a6) at (5.8,0) {};

        \node[bvertex] (b1) at (0.8,1) {};
        \node[bvertex] (b2) at (2,1) {};
        \node[bvertex] (b3) at (3.45,1) {};
        \node[bvertex] (b5) at (4.2,1.8) {};
        \node[bvertex] (b4) at (5,1) {};
        \draw[semithick] (a1)--(a2)--(a3)--(a1);
        \draw[semithick] (a4)--(a5)--(a6)--(a4);
        \draw[semithick] (a1)--(b1);
        \draw[semithick] (a3)--(b3)--(a4);
        \draw[semithick] (a2)--(b2)--(a1);
        \draw[semithick] (a4)--(b4)--(a6);
        \draw[semithick] (b1)--(b2);

        \draw[semithick] (b3)--(b4)--(b5)--(b3);
  \end{scope}

  \end{tikzpicture}
  \caption{A monopolar and a 2-subcolorable graph with colorings that give a~\abpartition. The vertices of~$A$ are black and the vertices of~$B$ are white. Left: a monopolar graph where~$A$ induces a cluster graph and~$B$ is an independent set. Right: a 2-subcolorable graph where~$A$ and~$B$ induce cluster graphs.}
\label{fig:mono-2sub}
\end{figure}
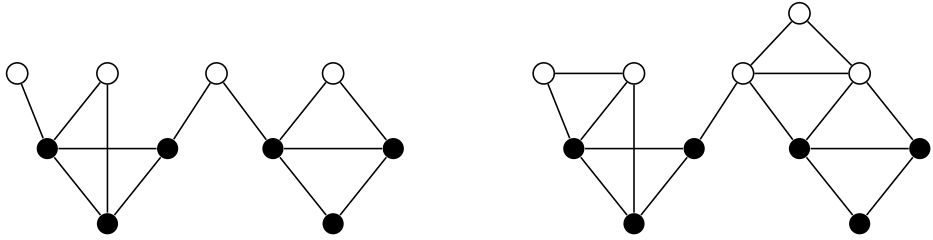
In contrast to the recognition problem of split graphs, which admits a linear-time algorithm~\cite{HS81},
\textsc{Monopolar Recognition} is NP-hard. This motivates a parameterized complexity analysis of \textsc{Monopolar Recognition}.
We consider the parameterized version of \textsc{Monopolar Recognition}, where the parameter $k$ is an upper bound on the number of clusters in $G[A]$, and use inductive recognition to show the following:

\begin{theorem}
  \label{thm:main:mono}
  \thmtxtmainmono
\end{theorem}%
\noindent The above result is a generalization of the linear-time algorithm for recognizing split graphs~\cite{HS81},
which we can obtain by setting~$k=1$.

We observe that inductive recognition, and the ideas used to obtain the result in Theorem~\ref{thm:main:mono}, can be generalized to \textsc{$(\Pi_{A},\Pi_{B})$-Recognition} problems in which properties
$\Pi_{A}$ and $\Pi_{B}$ satisfy certain conditions of \emph{mutual exclusivity}. To understand this notion of mutual exclusivity, consider the situation for split graphs.
Intuitively, recognizing split graphs is easy for the following reason: If we consider any bipartition of a split graph~$G$ into
a clique~$K$ and an independent set~$I$, then this bipartition differs only very slightly
from any other such bipartition of~$G$. This is because at most one vertex of~$K$ may be
part of any independent set of~$G$ and similarly, at most one vertex of~$I$ may be part of
any clique in~$G$. Thus, the two graph properties exclude each other to a large extent.
In the situation of monopolar graphs, the two
properties~$\Pi_A$ and~$\Pi_B$ defining monopolar graphs, are not mutually
exclusive:~$G[A]$ and~$G[B]$ could be an edgeless graph of arbitrary number of vertices. However, in the parameterized \textsc{Monopolar Recognition} problem, we restrict $G[A]$ to contain at most $k$ clusters. This
restriction makes the properties mutually exclusive, as~$G[A]$ may not contain an edgeless
graph on~$k + 1$ vertices anymore. Hence, any two monopolar partitions with $k$ clusters in $G[A]$ again differ only slightly if the parameter $k$ is small. We generalize this observation in the following definition:

\begin{definition}\label{def:mutually-exclusive}
  Two graph properties~$\Pi_A$ and~$\Pi_B$ are called \emph{mutually
    $d$-exclusive} if there is no graph with at least~$d$ vertices that
  fulfills~$\Pi_A$ and~$\Pi_B$.
\end{definition}

For a pair $(\Pi_{A},\Pi_{B})$ of mutually $d$-exclusive graph properties, we use inductive recognition to obtain the following result:
\begin{theorem}\label{thm:exclusive-rec}
  \thmtxtexclusiverec
\end{theorem}

Although Theorem~\ref{thm:exclusive-rec} is quite general, there are many natural cases of
\textsc{$(\Pi_{A},\Pi_{B})$-Recognition} to which it does not apply. Moreover, the degree of
the polynomial in the running time depends on~$d$, that is, the corresponding algorithm is an
XP algorithm for the parameter~$d$. Thus, while Theorem~\ref{thm:exclusive-rec} can be applied to solve the parameterized \textsc{Monopolar Recognition} problem, the running time obtained is not practical, even for moderate values of~$k$. In contrast, the much more efficient algorithm alluded to in Theorem~\ref{thm:main:mono} is tailored for parameterized \textsc{Monopolar Recognition}. An improvement from the XP algorithm implied by Theorem~\ref{thm:exclusive-rec} to an FPT algorithm for \emph{all} pairs of mutually exclusive properties is unlikely, as we show in Section~\ref{sec:XP}. Nevertheless, we show in Section~\ref{sec:gen} that the FPT algorithm for \textsc{Monopolar Recognition} can be adapted to work for certain pairs of mutually exclusive properties $\Pi_{A}$ and $\Pi_{B}$.

\paragraph*{2-Subcolorings}

Next, we study \emph{2-subcolorable graphs}; these are graphs for which the vertex set admits a bipartition into two cluster graphs~\cite{BFNW02}, and thus are a superset of unipolar graphs; an example is shown in Figure~\ref{fig:mono-2sub}. The recognition problem of 2-subcolorable graphs can be formulated as follows:
\begin{quote}
  \textsc{2-Subcoloring}\\
  \textbf{Input:} A graph $G=(V,E)$.\\
  \textbf{Question:} Does $G$ have a \emph{2-subcoloring}~$(A,B)$, that is, can~$V$ be partitioned into sets $A$ and $B$ such that each of $G[A]$ and $G[B]$ is a cluster graph?
\end{quote}

In contrast to the recognition problem of unipolar graphs, which admits a polynomial-time algorithm~\cite{unipolar0,unipolar,MY15}, \textsc{2-Subcoloring} is NP-hard~\cite{Ach97}. We consider \textsc{2-Subcoloring} parameterized by the number of clusters in $G[A]$, and use inductive recognition to show the following:

\begin{theorem}\label{thm:main:sub}
  \thmtxtmainsub
\end{theorem}
The above result can be seen as a generalization of the polynomial-time algorithms for recognizing unipolar graphs~\cite{unipolar0,unipolar,MY15}, which we can obtain by setting $k = 1$.
We remark that one faces various technical difficulties when designing algorithms for parameterized \textsc{2-Subcoloring} as it does not seem to yield to standard approaches in parameterized algorithms.
This is a testament to the power of the inductive recognition technique, which adds to the arsenal of existing techniques for designing parameterized algorithms. Observe also that Theorem~\ref{thm:exclusive-rec} does not apply to parameterized \textsc{2-Subcoloring}.

\paragraph*{Further Results}
We also consider the \textsc{2-Subcoloring} problem parameterized by a weaker parameter: the total number of clusters in $G[A]$ and $G[B]$. This parameterization makes the problem amenable to a branching strategy that branches on the placement of the endpoints of suitably-chosen edges and nonedges of the graph. In this way, we create partial 2-subcolorings $(A',B')$ where each vertex in $V \setminus (A' \cup B')$ is adjacent to the vertices of exactly two partial clusters, one in each of $G[A']$ and $G[B']$. Then we show that whether such a partial 2-subcoloring extends to an actual 2-subcoloring of $G$ can be tested in polynomial time via a reduction to \textsc{2-CNF-Sat}. We prove the following result:

\begin{theorem}\label{thm:main:sub2}
  \thmtxtmainsubtwo
\end{theorem}
Finally, we consider the  parameter consisting of the total number of vertices in $G[A]$. We observe that a straightforward branching strategy yields a generic fixed-parameter algorithm for many \textsc{$(\Pi_A,\Pi_B)$-Recognition} problems.

\begin{proposition}\label{prp:main:size}
  \prptxtmainsize
\end{proposition}

We complement our results by presenting in Section~\ref{sec:hardness} several hardness results showing that significant improvements over the algorithms presented in this paper for \textsc{Monopolar Recognition}, \textsc{2-Subcoloring}, and the \textsc{$(\Pi_{A},\Pi_{B})$-Recognition} problem in general, are unlikely.%

\paragraph*{Related Work}
As mentioned, split graphs and unipolar graphs can be recognized in linear time~\cite{HS81} and polynomial time~\cite{unipolar0,unipolar,MY15}, respectively. In contrast, if~$\Pi_A$ and~$\Pi_B$ can be characterized by a set of connected forbidden subgraphs and~$\Pi_A$ is not the set of all edgeless graphs, then~\textsc{$(\Pi_A,\Pi_B)$-Recognition} is NP-hard~\cite{Far04}. This implies the NP-hardness of \textsc{Monopolar Recognition} and \textsc{2-Subcoloring}. Up to the authors' knowledge, the parameterized complexity of \textsc{Monopolar Recognition} and \textsc{2-Subcoloring} has not been studied before. The known algorithms for both problems are not parameterized, and assume that the input graph belongs to a structured graph class; see~\cite{CH12,CH14,CH14b,EHSW+08,LN14} and~\cite{BFNW02,FJBS03,Gimbel2003,Sta08}, respectively.

Recently, Kolay and Panolan~\cite{KP15} considered the problem of deleting $k$ vertices or edges to obtain an $(r,\ell)$-graph. For integers $r,\ell$, a graph $G=(V,E)$ is an \emph{$(r,\ell)$-graph} if $V$ can be partitioned into $r$ independent sets and $\ell$ cliques. For example, $(2,0)$-graphs are precisely bipartite graphs and $(1,1)$-graphs are precisely split graphs. Observe that $(1,\cdot)$-graphs are \emph{not} necessarily monopolar, because monopolar graphs do not allow edges between the cliques (as $G[A]$ is a cluster graph), whereas such edges are allowed in $(1,\cdot)$\nobreakdash-graphs. These differences lead to substantially different algorithmic techniques. For example, since Kolay and Panolan~\cite{KP15} consider the deletion problem, they can use iterative compression in their work. Moreover, they consider $r,\ell < 3$, whereas the number of cliques~$k$ may be arbitrarily large in our setting. Nevertheless, during the development of our algorithms we were inspired by some of their observations. Using techniques similar to Kolay and Panolan~\cite{KP15}, \citet{KPRS16} also obtained fixed-parameter algorithms for deleting vertices in perfect graphs to obtain $(r, \ell)$-graphs.

Iterative localization~\cite{HKLRS13} is a technique that is similar to iterative recognition, which we were unaware of at the time when the journal version of this article~\cite{KKSL18} was published.
\citet{HKLRS13} introduced iterative localization and used it to develop fixed-parameter tractable algorithms for computing the cochromatic number\footnote{The \emph{cochromatic number} of a graph is the minimum number of colors needed in a coloring of the vertices of that graph such that each color class induces either an edgeless graph or a complete graph.} of permutation graphs, and the stabbing number of disjoint rectangles with axes-parallel lines. The core idea of iterative localization is to build a solution incrementally by iteratively adding entities (e.g.\ vertices) greedily, and then search the space around the obtained solution for an optimal one. For example, for a given a cochromatic coloring of a graph $G - v$, where $G$ is a perfect graph, a cochromatic coloring for $G$ can be obtained by simply adding $v$ with a new color. From this coloring, \citet{HKLRS13} then obtain an optimum coloring via a branching algorithm. Our formalization of iterative recognition captures iterative localization, and could be seen as a slight generalization of it, since we do not (necessarily) directly obtain a solution from an iteratively added vertex.

\paragraph*{Organization of the Article} After describing the necessary graph-theoretic notations and giving a brief background on parameterized complexity in Section~\ref{sec:prelim}, we introduce inductive recognition in Section~\ref{sec:inducregoc}. In Section~\ref{sec:XP}, we show our most general tractability result based on inductive recognition, the XP algorithm for mutually exclusive graph properties (Theorem~\ref{thm:exclusive-rec}). In Section~\ref{sec:mono}, we give the linear-time fixed-parameter algorithm for \textsc{Monopolar Recognition} parameterized by the number~$k$ of cliques (Theorem~\ref{thm:main:mono}). In Section~\ref{sec:gen}, we present more general graph classes such that we can obtain FPT algorithms for their corresponding recognition problems. In Section~\ref{sec:2-subcolor}, we turn to \textsc{2-Subcoloring} parameterized by the smaller number of cliques between the two parts, and give an FPT algorithm for this problem based on inductive recognition (Theorem~\ref{thm:main:sub}). In Section~\ref{sec:totalpolar}, we present FPT algorithms that are not based on inductive recognition, for recognition problems parameterized by weaker parameters (Proposition~\ref{prp:main:size} and Theorem~\ref{thm:main:sub2}). Our hardness results are presented in Section~\ref{sec:hardness}. In Section~\ref{sec:conclusion}, we summarize our findings and point out future research directions.

\section{Preliminaries}
\label{sec:prelim}
For~$n\in \mathbb{N}$, we let~$[n]:=\{1,\dots, n\}$ denote the integers from~1 through~$n$.
We follow standard graph-theoretic notation~\cite{Die12}. Let $G$ be a graph. By $V(G)$ and $E(G)$ we denote the vertex set and the edge set of $G$, respectively. The \emph{order} of a graph $G$ is $|V(G)|$. Throughout this work, let $n := |V(G)|$ and $m:=|E(G)|$. For $X \subseteq V(G)$, $G[X]$ denotes the \emph{subgraph of $G$ induced by $X$}. For a vertex $v \in G$, $N(v)$ and $N[v]$ denote the open neighborhood and the closed neighborhood of $v$, respectively. The {\em degree} of a vertex $v \in G$, denoted $\deg(v)$, is $|N(v)|$.  We say that a vertex $v$ is \emph{adjacent to a subset $X \subseteq V(G)$ of vertices} if $v$ is adjacent to at least one vertex in $X$.  For $X \subseteq V(G)$, we define $N(X) := (\bigcup_{v \in X}N(v))\setminus X$ and $N[X] := \bigcup_{v \in X}N[v]$, and for a family $\mathcal{X}$ of subsets $X \subseteq V(G)$, we define  $N(\mathcal{X}) := (\bigcup_{X \in \mathcal{X}}N(X)) \setminus (\bigcup_{X \in \mathcal{X}}X)$ and $N[\mathcal{X}] := \bigcup_{X \in \mathcal{X}}N[X]$. If $S$ is any set of vertices in $G$, we write $G-S$ for the subgraph of $G$ obtained by deleting all the vertices in $S$ and their incident edges. For a vertex $v \in V(G)$, we write $G-v$ for $G - \{v\}$.
By $P_3$ we
denote the graph that is a (simple) path on 3 vertices. We repeatedly use the following well-known fact:

\begin{fact} \label{fact:p3}
A graph is a cluster graph if and only if it does not contain $P_3$ as an induced subgraph.
\end{fact}

Let $r, s$ be positive integers. The \emph{Ramsey number} $R(r, s)$ is the smallest integer such that every graph of order at least $R(r, s)$ either contains a clique of $r$ vertices or an independent set of $s$ vertices.
Ramsey's theorem~\cite{ramsey} states that, for any two positive integers $r, s$, the number $R(r, s)$ exists. The following upper bound on $R(r, s)$ is known: $R(r, s) \leq {r+s-2 \choose r-1}$.

 A {\em parameterized problem} is a set of instances of the form $(x, k)$, where $x \in \Sigma^*$ for a finite alphabet set $\Sigma$, and
 $k \in \mathbb{N}$ is the {\em parameter}.
A parameterized problem $Q$ is {\it fixed-parameter tractable} (\FPT), if there exists an algorithm that on input $(x, k)$
decides if $(x, k)$ is a yes-instance of $Q$ in time $f(k)n^{\Oh(1)}$,
where $f$ is a computable function independent of $n = |x|$; an algorithm with this running time is called {\em \FPT~algorithm}.
A hierarchy of fixed-parameter intractability, {\it the \emph{W}-hierarchy} $\bigcup_{t
\geq 0} \W{t}$, was introduced based on the notion of {\em \FPT~reduction}, in which the $0$th level $\W{0}$ is the class \FPT.  It is commonly believed that $\W{1} \neq \mbox{\FPT}$. A parameterized problem $Q$ is in the parameterized complexity class \XP,  if there exists an algorithm that on input $(x, k)$
decides if $(x, k)$ is a yes-instance of $Q$ in time $\Oh(n^{f(k)})$,
where $f$ is a computable function independent of $n = |x|$. For more discussion on parameterized complexity, we refer to the literature~\cite{DF13,CFK+15}.

We make also use of the \emph{Exponential Time Hypothesis} (ETH). The ETH was formulated by Impagliazzo \etal~\cite{IPZ01}, and states that \textsc{$k$-CNF-Sat} (for any $k \geq 3$) cannot be solved in subexponential time $2^{o(n)}$, where $n$ is the number of variables in the input formula. Therefore, there exists a constant $c > 0$ such that \textsc{$k$-CNF-Sat} cannot be solved in time $\Oh(2^{cn})$. ETH has become a standard hypothesis in complexity theory for proving tight running time bounds results.

\section{Foundations of Inductive Recognition}\label{sec:inducregoc}
In this section, we describe the foundations of the general technique that we use to solve \textsc{$(\Pi_{A},\Pi_{B})$-Recognition} problems. The technique works in a similar way to the iterative compression technique by \citet{RSV04}.
Let ${\cal G}$ be an arbitrary hereditary graph class (that is, if $G \in {\cal G}$, then $G' \in {\cal G}$ for every induced subgraph $G'$ of $G$). We call an algorithm~${\cal A}$ an \emph{inductive recognizer} for~${\cal G}$ if given a graph $G=(V,E)$, a vertex $v \in V$ such that $G-v \in {\cal G}$, and a membership certificate for~$G-v\in {\cal G}$, algorithm~${\cal A}$ correctly decides whether $G \in {\cal G}$, and gives a membership certificate if~$G\in {\cal G}$. For example, in the case of recognizing monopolar graphs, the membership certificate may be a string encoding a monopolar partition.

\begin{theorem}\label{thm:in-rec}
Let ${\cal G}$ be an arbitrary hereditary graph class. Given an inductive recognizer ${\cal A}$ for ${\cal G}$, we can recognize whether a given graph $G=(V,E)$ is a member of ${\cal G}$ in time $\Oh(n+m) + \sum_{i=1}^{n} T(i)$, where $T(i)$  is the worst-case running time of ${\cal A}$ on a graph of order at most~$i$.
\end{theorem}
\begin{proof}
We first sort the vertices in an arbitrary order to obtain a list~$v_1,\ldots, v_n$.
 Let~$G_0 = (\emptyset, \emptyset)$ be the empty graph and~$G_i:=G[\{v_1,v_2,\dots,v_i\}]$, for $i=1,\ldots,n$. Since $\cal G$ is hereditary, $G_0$ is a member of ${\cal G}$ and we can easily compute a membership certificate of $G_0$ in ${\cal G}$. Then, for $i=1, \ldots, n$, in order, we run ${\cal A}$ on $(G_i, v_i)$ passing to ${\cal A}$ the certificate of membership of $G_{i-1}$ in ${\cal G}$, to decide whether  $G_i$ is a member of ${\cal G}$, and produce a membership certificate in case it is, but only if $G_{i-1}$ is a member of ${\cal G}$. If ${\cal A}$ decides that $G_{i}$ is not a member of ${\cal G}$ for some $i=1,\ldots,n$, then we answer that $G$ is not a member of ${\cal G}$; this is correct, because ${\cal G}$ is hereditary. Otherwise, we answer that $G$ is a member of ${\cal G}$; the correctness of this answer follows from the correctness of ${\cal A}$. The bound on the running time is straightforward.
\smartqed\end{proof}

For the purpose of this paper, we consider \emph{parameterized inductive recognizers} for~$(\Pi_A,\Pi_B)$-graphs. In addition to $G$ and $v$, these recognizers take a nonnegative integer $k$ as input. The above general theorem can then be instantiated as follows.

\begin{corollary} \label{cor:in-rec-k}
Let $k$ be a nonnegative integer, and let $\Pi_A$ and $\Pi_B$ be two hereditary graph properties. Let ${\cal G}_{k}$ be a class of $(\Pi_A, \Pi_B)$-graphs with an additional hereditary property that depends on~$k$.
Given a parameterized inductive recognizer ${\cal A}$ for ${\cal G}_{k}$, we can recognize whether a given graph $G=(V,E)$ is a member of ${\cal G}_{k}$ in time $\Oh(n+m) + \sum_{i=1}^{n} T(i,k)$, where $T(i,k)$ is the worst-case running time of ${\cal A}$ with parameter~$k$ on a graph of order at most~$i$.
\end{corollary}

\begin{corollary} \label{cor:in-rec}
Let $k$ be a nonnegative integer and let $\Pi_A$ and $\Pi_B$ be two hereditary graph properties. Let ${\cal G}_{k}$ be a class of $(\Pi_A, \Pi_B)$-graphs with an additional hereditary property that depends on $k$. Given a parameterized inductive recognizer ${\cal A}$ for ${\cal G}_{k}$ that runs in time $f(k) \cdot \Delta$, where $\Delta$ is the maximum degree of the input graph and $f$ is an arbitrary computable function, we can recognize whether a given graph $G=(V,E)$ is a member of ${\cal G}_{k}$ in time $f(k) \cdot \Oh(n+m)$.
\end{corollary}
\begin{proof}[Proof of Corollary~\ref{cor:in-rec}]
  Modify the algorithm in the proof of Theorem~\ref{thm:in-rec} by, instead of sorting the vertices~$v_1, \ldots, v_n$ arbitrarily, sorting them in nondecreasing order of their vertex-degree in~$\Oh(n+m)$ time (using Bucket Sort for example). In this way, for each $i \in [n]$, graph~$G_i$ has its maximum degree upper bounded by the degree of vertex~$v_i$. Thus, if the running time of ${\cal A}$ depends linearly on the maximum degree of the input graph and arbitrarily on $k$, then the total running time is linear for every fixed $k$.
\smartqed\end{proof}

\section{A General Application of Inductive Recognition}\label{sec:XP} Recall that
\textsc{$(\Pi_{A},\Pi_{B})$-Recognition} is NP-hard if~$\Pi_A$ and~$\Pi_B$ can be
characterized by a set of connected forbidden induced subgraphs and, additionally,~$\Pi_A$
is not the set of all edgeless graphs~\cite{Far04}. While being quite general, this hardness result is
not exhaustive and ideally, we would like to obtain a complexity dichotomy for
\textsc{$(\Pi_{A},\Pi_{B})$-Recognition}. As a first application of inductive recognition and a step towards such a complexity dichotomy, we show how inductive recognition can be used to solve \textsc{$(\Pi_{A},\Pi_{B})$-Recognition} for hereditary mutually $d$-exclusive graph properties~$\Pi_A$,~$\Pi_B$, as defined in Definition~\ref{def:mutually-exclusive} (Section~\ref{sec:intro}).

To apply inductive recognition, we need to describe an inductive recognizer
for~$(\Pi_A,\Pi_B)$-graphs, that is, we need to give an algorithm for the following
problem.
\begin{quote}
  \textsc{Inductive $(\Pi_{A},\Pi_{B})$-Recognition}\\
  \textbf{Input:} A graph $G=(V,E)$, a vertex~$v\in V$, and a partition~$(A',B')$ of~$G'=G-v$ such that~$G[A']\in \Pi_{A}$ and~$G[B']\in \Pi_{B}$.\\
  \textbf{Question:} Does $V$ have a~\abpartition~$(A,B)$, that is, a partition such that~$G[A]\in \Pi_{A}$ and~$G[B]\in \Pi_{B}$?
\end{quote}

For mutually~$d$-exclusive graph properties, we can solve this
problem by starting a search from the partition~$(A',B')$
of~$G-v$. Herein, we can use the fact that the number of vertices that
can be moved from~$A'$ to~$B$ and from $B'$~to~$A$ are each upper-bounded by
$d$, because~$G[A']\in \Pi_A$
and~$G[B']\in \Pi_B$. This implies that the partition~$(A,B)$ is
determined to a large extent by the partition~$(A',B')$.

\begin{lemma}\label{lem:exclusive-rec}
  If~$\Pi_{A}$ and~$\Pi_B$ are hereditary and mutually~$d$-exclusive graph properties, and membership of~$\Pi_A$
  and~$\Pi_B$ can be decided in polynomial time, then \textsc{Inductive
    $(\Pi_{A},\Pi_{B})$-Recognition} can be solved in time $n^{2d + \Oh(1)}$.
\end{lemma}
\begin{proof}
  Assume there is a~\abpartition~$(A,B)$ of~$V$. Since~$\Pi_A$ and~$\Pi_B$ are mutually~$d$-exclusive, at most~$d-1$
  vertices of~$A'$ are contained in~$B$ and at most~$d-1$ vertices of~$B'$ are contained
  in~$A$. Consequently, we can decide whether~$G$ is a $(\Pi_A,\Pi_B)$-graph by the following
  algorithm. Consider each pair of~$\tilde{A}\subseteq A'$ and~$\tilde{B}\subseteq B'$
  such that~$|\tilde{A}|< d$ and~$|\tilde{B}|< d$. Determine whether $$((A'\cup
  \{v\}\cup \tilde{B})\setminus \tilde{A}, (B'\cup \tilde{A})\setminus \tilde{B})$$ is
  a~\abpartition{} and output it if this is the case. Otherwise, check whether~$$((A'\cup
  \tilde{B})\setminus \tilde{A}, (B\cup \{v\} \cup \tilde{A})\setminus \tilde{B})$$ is
  a~\abpartition{} and output it if this is the case. If both tests fail for all pairs of~$\tilde{A}$
  and~$\tilde{B}$, then output that~$G$ is not
  a~$(\Pi_A,\Pi_B)$-graph.

  The correctness follows from the fact that the algorithm considers both possibilities
  of placing~$v$ in $A$ or $B$, and all possibilities for moving vertices from~$A'$ to~$B$ and from~$B'$
  to~$A$. The running time is~$n^{2d-2}\cdot n^{\Oh(1)}=n^{2d + \Oh(1)}$, since we consider
  altogether at most~$2\cdot n^{2d-2}$ different bipartitions, and for each
  bipartition the membership of the two parts in~$\Pi_A$ and~$\Pi_B$ can be determined
  in polynomial time.
\end{proof}
By combining Theorem~\ref{thm:in-rec} and Lemma~\ref{lem:exclusive-rec}, we
immediately obtain the following.
\begin{reptheorem}{thm:exclusive-rec}
  \thmtxtexclusiverec
\end{reptheorem}
Two hereditary graph properties~$\Pi_A, \Pi_B$ are mutually~$d$-exclusive for some integer~$d$ if and only if~$\Pi_A$ excludes some edgeless
graph and~$\Pi_B$ excludes some clique, or vice versa: Clearly, the ``only if''-part holds. For the ``if''-part, we obtain the
following upper bounds on~$d$.
\begin{proposition}\label{fact:mutually-exclusive}Let~$\Pi_A$ and~$\Pi_B$ be hereditary graph properties.
  If~$\Pi_A$ excludes an edgeless graph of order~$s_a$ and~$\Pi_B$ excludes a complete graph of order~$s_b$,
  then $\Pi_{A}$ and~$\Pi_{B}$ are mutually~$R(s_a,s_b)$-exclusive.
\end{proposition}
\begin{proof}
  By the definition of Ramsey numbers, every graph of order~$R(s_a,s_b)$ contains either
  an edgeless subgraph of order at least~$s_a$, or a complete subgraph of
  order~$s_b$. Thus, every graph of order at least~$R(s_a,s_b)$ fulfilling~$\Pi_A$
  contains a complete graph of order~$s_b$ and thus does not fulfill~$\Pi_B$. Since Ramsey
  numbers are symmetric, every graph of order~$R(s_a,s_b)$ fulfilling~$\Pi_B$ contains an
  edgeless subgraph of order~$s_a$ and thus does not fulfill~$\Pi_A$.
\end{proof}
   If~$\Pi_A$ and~$\Pi_B$ fulfill the conditions of the above proposition and are recognizable in polynomial time, then~Theorem~\ref{thm:exclusive-rec} applies.
\begin{corollary}\label{cor:sparse-dense}
  Let~$\Pi_A$ and~$\Pi_B$ be hereditary graph properties such that
  membership of~$\Pi_A$ and~$\Pi_B$ can be decided in polynomial time.  If~$\Pi_A$
  excludes a fixed edgeless graph and~$\Pi_B$ excludes a fixed complete graph, then
  \textsc{$(\Pi_{A},\Pi_{B})$-Recognition} can be solved in polynomial time.
\end{corollary}
Observe that if~$\Pi_A$ and~$\Pi_B$ both contain arbitrarily large edgeless graphs, then~$\Pi_A$ and~$\Pi_B$ are not mutually exclusive. Similarly, $\Pi_A$ and~$\Pi_B$ are not mutually exclusive if
both contain arbitrarily large complete graphs. As a consequence, Corollary~\ref{cor:sparse-dense} summarizes the applications of Theorem~\ref{thm:exclusive-rec}. A natural question is whether in Theorem~\ref{thm:exclusive-rec} the dependency of the running time on $d$ can be improved. A substantial improvement to $f(d) \cdot n^{\Oh(1)}$, however, is unlikely as we show in the remainder of this section.

\paragraph{A Note on Vertex Deletion Problems}
Theorem~\ref{thm:exclusive-rec} has some applications for vertex deletion problems in
undirected graphs which we illustrate with an example below. On the negative side, this example also gives a graph property $\Pi_A$ and a family~$\mathcal{F}$ of graph properties such that $\Pi_A$ and each $\Pi_B\in \mathcal{F}$ are mutually $d(\Pi_B)$-exclusive and \textsc{$(\Pi_{A},\Pi_B)$-Recognition} is W[1]-hard with respect to $d(\Pi_B)$ if we consider~$\Pi_B$ as part of the input.

Consider the \textsc{Vertex Cover} problem where we are
given a graph~$G$ and want to determine whether~$G$ has a vertex cover~$S$ of size at
most~$k$, that is, whether at most~$k$ vertices of~$G$ can be deleted such that the
remaining graph is edgeless. This problem can be phrased as
a~\textsc{$(\Pi_{A},\Pi_{B})$-Recognition} problem: $\Pi_A$ is the class of edgeless
graphs and~$\Pi_B$ is the class of graphs of order at most~$k$.

\newcommand{\pDVC}{\textsc{Dense Vertex Cover}}
The standard parameter for \textsc{Vertex Cover} is
the solution size~$k$. Let us consider instead the smaller parameter~$\ell$,
the ``size of the maximum independent set over all
size-$k$ solutions~$S$''. That is, we modify the problem by adding an
additional integer~$\ell$ to the input and we want to decide whether there is a vertex cover $S$ of size at most~$k$ such that the size of a maximum independent set in~$G[S]$ is $\ell$. Call this problem \pDVC. Clearly, $\ell$ can be arbitrarily smaller than~$k$. A given instance $(G, k, \ell)$ of \pDVC\ can again be formulated in terms of~\textsc{$(\Pi_{A},\Pi_{B})$-Recognition}: As before, $\Pi_A$ is the set of edgeless graphs, and $\Pi_B$ is now the class of graphs of order at most~$k$ which have no independent set of size~$\ell+1$. Also, since~$\Pi_A$ excludes the complete graph on two vertices, $\Pi_{A}$ and $\Pi_{B}$ are mutually $(\ell + 1)$-exclusive. Hence, Theorem~\ref{thm:exclusive-rec}
implies for every fixed~$\ell$ a polynomial-time algorithm for \pDVC, in other words, an XP algorithm for the parameter~$\ell$.  

Ideally, we would like to replace this XP~algorithm by an FPT~algorithm. This, however, is unlikely, as the following proposition shows.
\begin{proposition}\label{prop:dense-vc-hard}
  \pDVC\ parameterized by~$\ell$ is W[1]-hard.
\end{proposition}
\begin{proof}[Proof]
  Consider the \textsc{Partitioned Independent Set} problem where we
  are given a graph~$G=(V,E)$ with a vertex partition~$V_1,\ldots,V_t$
  such that each part~$V_i$ induces a clique in~$G$ and the task is to decide
  whether~$G$ has an independent set of size~$t$. Equivalently, we may
  ask for a vertex cover of size~$n-t$, giving a trivial reduction from \textsc{Partitioned Independent Set} to \textsc{Dense Vertex Cover}. This reduction is parameter-preserving: the input graph is not changed by the reduction and thus~$\ell\le t$, because the maximum independent set size in~$G$ is at most~$t$.
  Hence, it suffices to establish that \textsc{Partitioned
    Independent Set} is W[1]-hard parameterized by
  the size~$t$ of the desired independent set.

  The W[1]-hardness of \textsc{Partitioned
    Independent Set} can be seen by a folklore reduction from \textsc{Independent Set} parameterized by the size of the independent set, which is well known to be W[1]-hard (see e.g.~\cite{DF13,CFK+15}); for the sake of completeness, we give a short description. Let $(G,k)$ be an instance of \textsc{Independent Set} with an $n$-vertex graph $G=(V,E)$ and integer $k$. Create $k$ copies $v_1,\ldots,v_k$ of each vertex $v \in V$, and let $V' = \{v_1,\ldots,v_k \mid v \in V\}$ and $V'_i = \{v_i \mid v \in V\}$. Let $G'$ be the graph with vertex set $V'$ where $u_i$ and $v_j$ are adjacent if and only if $i=j$, or $i \not= j$ and $u \in N_G[v]$. Observe that $G'[V'_i]$ is a clique for each $i$. Moreover, $G$ has an independent set of size $k$ if and only if $G'$ has an independent set of size $k$. Hence, by setting $t = k$, we complete the reduction.
\end{proof}

We have noted above that each instance of \pDVC\ can be solved by a
call to the algorithm in Theorem~\ref{thm:exclusive-rec} for two
mutually $(\ell +1)$-exclusive graph properties. Combining this with Proposition~\ref{prop:dense-vc-hard}, we obtain the following corollary.
\begin{corollary}
  Unless FPT${}={}$W[1], %
  the running time in Theorem~\ref{thm:exclusive-rec} cannot be improved to $f(d) \cdot n^{\Oh(1)}$.
\end{corollary}

Nevertheless, Theorem~\ref{thm:exclusive-rec} and Corollary~\ref{cor:sparse-dense} imply XP~algorithms for many vertex deletion problems when the parameter is the size of the
independent set of the solution; examples of such applications are \textsc{Feedback Vertex Set} (where the remaining graph is a forest), \textsc{Bounded Degree Deletion} (where the remaining graph has degree at most~$r$ for some constant~$r$), and \textsc{Planar Vertex Deletion}.

\section{An FPT Algorithm for \textsc{Monopolar Recognition}}
\label{sec:mono}

\textsc{Monopolar Recognition} is the special case of~\textsc{$(\Pi_A,\Pi_B)$-Recognition}
where $\Pi_A$ is the set of cluster graphs and~$\Pi_B$ is the set of edgeless
graphs. Here, we consider \textsc{Monopolar Recognition} with the number~$k$ of clusters as a
parameter. This further restricts~$\Pi_A$: by bounding the parameter~$k$ we
constrain~$\Pi_A$ to be the set of cluster graphs with at most~$k$ clusters. Thus, the
graphs in~$\Pi_A$ cannot contain an edgeless graph of order~$k+1$ as subgraph. In other words, every graph in~$\Pi_A$ that has order at least~$k+1$ contains at least one edge. Altogether this implies the following.
\begin{fact}
  If~$\Pi_A$ is the set of cluster graphs with at most~$k$ clusters and~$\Pi_B$ is the set of edgeless graphs, then~$\Pi_A$ and~$\Pi_B$ are mutually~$(k+1)$-exclusive.
\end{fact}

Thus,~$\Pi_A$ and~$\Pi_B$ fulfill the conditions of Theorem~\ref{thm:exclusive-rec} for
each fixed~$k$ which implies an XP algorithm for
\textsc{Monopolar Recognition} parameterized by~$k$. In this section, we give a
linear-time FPT~algorithm for \textsc{Monopolar Recognition} parameterized by the number
of clusters~$k$.

Throughout, given a graph $G=(V,E)$ and a nonnegative integer $k$, we say that a monopolar partition~$(A,B)$ of~$G$ is
\emph{valid} if~$G[A]$ is a cluster graph with at most~$k$ clusters. Using Corollary~\ref{cor:in-rec}, it suffices to give a parameterized inductive recognizer for graphs with a valid monopolar partition. That is, we need to solve the following problem in time $f(k) \cdot \Delta$, where $f$ is some computable function and $\Delta$ is the maximum degree of~$G$:

\begin{quote}
  \textsc{Inductive Monopolar Recognition}\\
  \textbf{Input:} A graph $G=(V,E)$, a vertex~$v\in V$, and a valid monopolar partition~$(A',B')$ of~$G'=G-v$.\\
  \textbf{Question:} Does $G$ have a valid monopolar partition~$(A,B)$?
\end{quote}
In the following, we fix an instance of~\textsc{Inductive
  Monopolar Recognition} with a graph~$G=(V,E)$, a vertex~$v\in V$, and a valid
monopolar partition~$(A',B')$ of~$G'=G-v$.

\looseness=-1
To find a valid
monopolar partition~$(A,B)$ of~$G$, we try the
two possibilities of placing~$v$ in~$A$ or placing~$v$ in~$B$. More precisely,
in the first case, we start a search from the bipartition~$(A'\cup
\{v\},B')$, and in the second case, we start a search from the
bipartition~$(A',B' \cup \{v\})$. Neither of these two partitions
is necessarily a valid monopolar partition of~$G$. The search strategy
is to try to ``repair'' a candidate partition by moving few vertices from one
part of the partition to the other part. During this process, if a vertex is moved from one
part to the other, then it will never be moved back. To formalize this approach, we introduce the notion of constraints.

\begin{definition} \rm A \emph{constraint}~$\C=(\sA,\pA,\sB,\pB)$ is a four-partition of~$V$ such
  that~$\sA\subseteq A'$ and~$\sB\subseteq B'$. The vertices in~$\pA$ and~$\pB$ are called \emph{permanent} vertices of the
  constraint.
  A constraint $\C=(\sA,\pA,\sB,\pB)$ is \emph{fulfilled} by a vertex bipartition~$(A,B)$ of~$G$
   if $(A, B)$ is a valid monopolar partition of $G$ such that $\pA\subseteq A$ and $\pB\subseteq B$.
\end{definition}

\noindent The permanent vertices in~$\pA$ and~$\pB$ in the above definition
will correspond to those vertices that were moved during the
search from one part to the other part. The following observation is straightforward:

\begin{fact}\label{fact:init} Each valid monopolar partition~$(A,B)$ of~$G$
  fulfills either $(A',\{v\},B',\emptyset)$ or $(A',\emptyset,B',\{v\})$.
\end{fact}
We call the two constraints in Fact~\ref{fact:init} the \emph{initial constraints} of the
search. We solve~\textsc{Inductive Monopolar Recognition} by
giving a search-tree algorithm that determines for each of the two initial constraints whether there is a partition fulfilling it.
The root of the search tree is a dummy node that has two children, associated with the two initial constraints.
Each non-root node in the search tree
is associated with a constraint $\C$, and the algorithm
searches for a solution that fulfills $\C$.
To this end, the algorithm applies reduction and branching rules that find vertices that in \emph{every} valid monopolar partition~$(A,B)$
fulfilling $\C$ are in $\sA \cap B$ or $\sB \cap A$; that is, these vertices must `switch sides'.

Formally, a \emph{reduction rule} that is applied to a constraint~$\C$ associated with a node $\alpha$ in the search tree associates~$\alpha$ with a new constraint~$\C'$ or rejects $\C$; the reduction rule is
\emph{correct} either if $\C$ has a fulfilling partition if and only if~$\C'$
does, or if the rule rejects $\C$, then no valid monopolar partition of $G$ fulfills $\C$. A \emph{branching rule} applied to a constraint~$\C$ associated with a node $\alpha$ in the search tree produces
more than one child node of~$\alpha$, each associated with a constraint; the branching rule is
\emph{correct} if $\C$ has a fulfilling partition if and only if at
least one of the child nodes of~$\alpha$ is associated with a constraint $\C'$ that has
a fulfilling partition.

The algorithm first performs the reduction rules exhaustively, in order, and then performs the branching rules, in order. That is, Reduction Rule~$i$ may only be applied if  Reduction Rule~$i'$ for all $i' < i$ cannot be applied. In particular, after Reduction Rule~$i$ is applied, we start over and apply Reduction Rule~$1$, and so on. The same principle applies to the branching rules; moreover, branching rules are only applied if no reduction rule can be applied.

Let $\C=(\sA,\pA,\sB,\pB)$ be a constraint. We now describe the reduction rules. Throughout, recall from Fact~\ref{fact:p3} that cluster graphs contain no $P_3$ as an induced subgraph.
The first reduction rule identifies obvious cases in which a constraint cannot
be fulfilled.

\begin{rrule}\label{rr:per-verts-m} %
  If~$G[\pA]$ is not a cluster graph with at most $k$ clusters, or if~$G[\pB]$ is not an edgeless graph,
  then reject the current constraint.
\end{rrule}
\begin{proof}[Proof of correctness]
If~$G[\pA]$ is not a cluster graph with at most $k$ clusters, then there is no
  valid monopolar partition~$(A,B)$ satisfying~$\pA\subseteq
  A$. Similarly, there is no valid monopolar partition~$(A,B)$ satisfying~$\pB\subseteq B$ if~$G[\pB]$ is not an edgeless graph. \smartqed
\end{proof}
The second reduction rule finds vertices that must be moved from~$\sB$ to~$\pA$.

\begin{rrule}\label{rr:half-per-edge}
  If there is a vertex~$u\in \sB$ that has a neighbor in~$\pB$, then set $\pA \leftarrow \pA\cup \{u\}$ and $\sB \leftarrow \sB\setminus \{u\}$; that is, replace $\C$ with the constraint~$(\sA,\pA\cup \{u\},\sB\setminus \{u\},\pB)$.
\end{rrule}
\begin{proof}[Proof of correctness]
  For every partition~$(A,B)$ fulfilling~$\C$, $G[B]$~is an
  edgeless graph and~$\pB\subseteq B$. Hence,~$u\in A$.\smartqed
\end{proof}

The third reduction rule finds vertices that must be moved from~$\sA$ to~$\pB$.

\begin{rrule}\label{rr:2-per-p3}
  If there is a vertex~$u\in \sA$ and two vertices~$w,
  x\in \pA$ such that~$G[\{u,w,x\}]$ is a~$P_3$, set $\sA \leftarrow \sA\setminus \{u\}$ and $\pB \leftarrow \pB\cup \{u\}$.
\end{rrule}
\begin{proof}[Proof of correctness]%
  For every partition~$(A,B)$ fulfilling~$\C$, the graph~$G[A]$ is a cluster
  graph and~$\pA\subseteq A$. Hence,~$u\in B$. \smartqed
\end{proof}
The first branching rule identifies pairs of vertices from~$\sA$ such that at least one of them must be moved to~$\pB$ because they form a~$P_3$ with a vertex in~$\pA$.

\begin{brule}\label{br:1-per-p3}
  If there are two vertices~$u,w\in \sA$ and a vertex~$x\in \pA$ such that~$G[\{u,w,x\}]$ is a~$P_3$, then branch into two branches: one associated with the constraint~$(\sA\setminus\{u\},\pA,\sB,\pB\cup \{u\})$ and one associated with the constraint~$(\sA\setminus \{w\}, \pA, \sB, \pB\cup \{w\})$.
\end{brule}
\begin{proof}[Proof of correctness]%
  For every partition~$(A,B)$ fulfilling~$\C$,~$G[A]$ is a cluster
  graph and~$\pA\subseteq A$. Hence,~$u\in B$ or~$w\in B$. \smartqed
\end{proof}
It is important to observe that if none of the previous rules applies, then~$(\sA\cup \pA,\sB\cup\pB)$ is a monopolar partition (we prove this rigorously in Lemma~\ref{lem:reduced-instance}). However, $G[\sA\cup \pA]$ may consist of too many
clusters for this to be a \emph{valid} monopolar partition.
To check whether it is possible to reduce the number of clusters in $G[\sA\cup \pA]$, we apply a second branching rule that deals with singleton clusters in~$G[A']$.

\begin{brule}\label{br:singleton-cluster}
  If there is a vertex~$u\in \sA$ such that~$\{u\}$ is a cluster in~$G[A']$, then branch into two branches: the first is associated with the constraint~$(\sA\setminus \{u\},\pA\cup \{u\},\sB,\pB)$, and the second is associated with the constraint~$(\sA\setminus \{u\},\pA,\sB,\pB\cup \{u\})$.
\end{brule}
\begin{proof}[Proof of correctness]%
  For every partition~$(A,B)$ fulfilling~$\C$, we have either~$u\in A$ or~$u\in B$. \smartqed
\end{proof}
If no more rules apply to a constraint~$\C$, then we can determine whether~$\C$ can be fulfilled:

\begin{lemma}\label{lem:reduced-instance}
  Let $\C=(\sA,\pA,\sB,\pB)$ be a constraint to which
  \rrrefm{rr:per-verts-m}, \ref{rr:half-per-edge}, and \ref{rr:2-per-p3},
  and~\brrefm{br:1-per-p3} and \ref{br:singleton-cluster} do not apply. Then
  $(\sA\cup\pA,\sB\cup\pB)$ is a monopolar partition. Moreover, there
  is a valid monopolar partition~$(A,B)$ fulfilling~$\C$ if and only
  if $(\sA\cup\pA,\sB\cup\pB)$ is valid.
\end{lemma}
\begin{proof}
  First, we show that $(\sA\cup\pA,\sB\cup\pB)$ is a monopolar
  partition. There are no induced $P_3$'s in~$G[\sA\cup\pA]$,
  because~\rrrefm{rr:per-verts-m} and~\ref{rr:2-per-p3} and~\brref{br:1-per-p3} do
  not apply, and because there are no induced $P_3$'s in~$G$ containing
  three vertices from~$\sA\subseteq A'$. Similarly, there are no edges
  in~$G[\sB\cup\pB]$, because~\rrrefm{rr:per-verts-m} and~\ref{rr:half-per-edge} do not
  apply, and because there are no edges in~$G[B']$ and~$\sB\subseteq B'$.

  To show the second statement in the lemma, observe
  that, if~$(\sA\cup\pA,\sB\cup\pB)$ is valid, then~$\C$ is fulfilled by $(\sA\cup\pA,\sB\cup\pB)$.  %
It remains to show that, %
if~$(\sA\cup\pA,\sB\cup\pB)$ is not valid, then each monopolar partition $(A,B)$ of~$G$ fulfilling~$\C$ is not valid. For the sake of contradiction, assume that this is not the case and let~$(A,B)$ be a valid monopolar partition fulfilling~$\C$.  Since~$(\sA\cup\pA,\sB\cup\pB)$ is a monopolar partition of $G$ that is not valid, there are more than~$k$ clusters in $G[\sA\cup\pA]$. Thus, there is a cluster~$Q$ in~$G[\sA\cup\pA]$ such that $Q \subseteq B$. Note that~$|Q| = 1$, because $G[B]$ has no edges and $Q \subseteq B$. Because~$(A,B)$ fulfills~$\C$, $Q \cap \pA = \emptyset$ and thus $Q \subseteq \sA$. Hence, $Q$ is a subset of a cluster $Q'$ of $G[A']$, as~$Q\subseteq\sA\subseteq A'$.  However,~$|Q'| \geq 2$, because~\brref{br:singleton-cluster} does not apply even though $Q \subseteq \sA$. Hence, any rule that moved the vertices of $Q' \setminus Q$ was not~\brref{br:singleton-cluster}. Then the description of the other rules implies that $Q' \setminus Q \subseteq \pB$. Note that $\pB \subseteq B$, because~$(A,B)$ fulfills~$\C$. Hence,~$Q' \subseteq B$ and thus $G[B]$ contains an edge. %
Therefore,~$(A,B)$ is not a monopolar partition, a contradiction to our choice of~$(A,B)$.\smartqed
\end{proof}
The following lemmas will be used to upper bound the depth of the search tree, and the number of applications of each rule along each root-leaf path in this tree.
Herein a \emph{leaf} of the search tree is a node associated either with a constraint that~\rrref{rr:per-verts-m} rejects, or with a constraint to which no rule applies.

\begin{lemma}\label{lem:few-up}
   Along any root-leaf path in the search tree of the algorithm,~\rrref{rr:half-per-edge} is applied at most $k+1$ times.
\end{lemma}

\begin{proof}
  Let~$\C=(\sA,\pA,\sB,\pB)$ be a constraint obtained from an initial
  constraint via~$k+1$ applications of
  \rrref{rr:half-per-edge} and an arbitrary number of applications of~\rrrefm{rr:per-verts-m} and~\ref{rr:2-per-p3}, and~\brrefm{br:1-per-p3} and~\ref{br:singleton-cluster}.
  Each application of~\rrref{rr:half-per-edge} adds a vertex of~$B'$
  to~$\pA$. Since~$G[B']$ is edgeless, any monopolar
  partition~$(A,B)$ with~$\pA\subseteq A$ has at least~$k+1$ clusters
  in~$G[A]$ and, therefore, is not valid. \rrref{rr:per-verts-m} will then be applied before any further application of \rrref{rr:half-per-edge}, and the constraint $\C$ will be rejected. \smartqed
\end{proof}

\begin{lemma}\label{lem:few-down}
  Along any root-leaf path in the search tree of the algorithm,~\rrref{rr:2-per-p3} and~\brrefm{br:1-per-p3} and~\ref{br:singleton-cluster} are applied at most $k+1$ times in total.
\end{lemma}
\begin{proof}
  Let~$\C=(\sA,\pA,\sB,\pB)$ be a constraint obtained from an initial
  constraint via~$k+1$ applications
  of~\rrref{rr:2-per-p3} and \brrefm{br:1-per-p3} and~\ref{br:singleton-cluster}, and
  an arbitrary number of applications of the other rules.
  Let~$k_s$ denote the number of singleton clusters
  in~$G[A']$.
  Observe that each application of~\rrref{rr:2-per-p3} or~\brrefm{br:1-per-p3} and~\ref{br:singleton-cluster} makes a vertex of $\sA \subseteq A'$ permanent by placing it in $\pA$ or~$\pB$. By the description of all rules, a vertex will never be made permanent twice. Hence, out of the~$k+1$ applications
  of~\rrref{rr:2-per-p3} and \brrefm{br:1-per-p3} and~\ref{br:singleton-cluster}, at most $k_s$ make the vertex from a singleton cluster of~$G[A']$ permanent. Observe that~\brref{br:singleton-cluster}  cannot make a vertex from a nonsingleton cluster in~$G[A']$ permanent. Thus, \rrref{rr:2-per-p3} and~\brref{br:1-per-p3} make at least~$k-k_s+1$ vertices in the~$k-k_s$
  nonsingleton clusters of~$G[A']$ permanent.
Since $k-k_s+1 \geq 1$, this also implies that a nonsingleton cluster exists.
  By the pigeonhole principle, out of the $k-k_s + 1$ vertices that are made permanent by \rrref{rr:2-per-p3} and~\brref{br:1-per-p3}, two are from the same nonsingleton cluster in~$G[A']$. Since both~\rrref{rr:2-per-p3} and~\brref{br:1-per-p3} only move vertices from~$\sA$ to~$\pB$, it follows that~$\pB$ contains two
  adjacent vertices. Then the constraint $\C$ will be rejected by~\rrref{rr:per-verts-m}, which is applied before any further rule is applied. \smartqed
\end{proof}

\begin{theorem}\label{thm:rt-inductive-m}
  \textsc{Inductive Monopolar Recognition} can be solved in~$\Oh(2^k\cdot
  k^3\cdot \Delta)$ time, where~$\Delta$ is the maximum degree
  of~$G$.
\end{theorem}

\begin{proof}
We call a leaf of the search tree associated with a constraint to which no rule applies an \emph{exhausted leaf}.
By Lemma~\ref{lem:reduced-instance} and the correctness of the rules,~$G$ has a valid monopolar
partition if and only if for at least one exhausted leaf node, the partition $(\sA\cup\pA,\sB\cup\pB)$, induced by the constraint~$\C$ associated with that node, is a valid monopolar partition.  Hence, if the search tree has an exhausted leaf for which the partition $(\sA\cup\pA,\sB\cup\pB)$, induced by the constraint~$\C$ associated with that node, is a valid monopolar partition, the algorithm answers `yes'; otherwise, it answers `no'.
Therefore, the described search-tree algorithm correctly decides an instance of \textsc{Inductive Monopolar Recognition}.

To upper bound the running time, let ${\cal T}$ denote the search tree of the algorithm. By Lemma~\ref{lem:few-down}, \brrefm{br:1-per-p3} and~\ref{br:singleton-cluster} are applied at most $k+1$ times in total along any root-leaf path in ${\cal T}$. It follows that the depth of ${\cal T}$ is at most $k+2$. As each of the branching rules is a two-way branch, ${\cal T}$ is a binary tree, and thus the number of leaves in ${\cal T}$ is $\Oh(2^k)$.

\looseness=-1
The running time along any root-leaf path in ${\cal T}$ is dominated by the overall time taken along the path to test the applicability of the reduction and branching rules, and to apply them.
    By Lemma~\ref{lem:few-up} and Lemma~\ref{lem:few-down}, along any root-leaf path in ${\cal T}$ the total number of applications of~\rrrefm{rr:half-per-edge} and~\ref{rr:2-per-p3} and~\brrefm{br:1-per-p3} and~\ref{br:singleton-cluster} is $\Oh(k)$. \rrref{rr:per-verts-m} is applied once before the application of each of the aforementioned rules. It follows that the total number of applications of all rules along any root-leaf path in ${\cal T}$ is $\Oh(k)$. Moreover, ${\cal T}$ has $\Oh(2^k)$ leaves as argued before. Therefore, we test for the applicability of the rules and apply them, or use the check of Lemma~\ref{lem:reduced-instance}, at most $\Oh(2^k\cdot k)$ times.

We now upper bound the time to test the applicability of the rules and to apply them by $\Oh(k^2 \cdot \Delta)$.
   Let~$\C=(\sA,\pA,\sB,\pB)$ be a constraint associated with a node in ${\cal T}$. Observe that each cluster in~$G[\sA]$ has
  size~$\Oh(\Delta)$. Since~$G[\sA]$ has at most~$k$ clusters, this
  implies that~$|\sA|\leq k\cdot \Delta$. Thus, in~$\Oh(k\cdot \Delta)$ time,
  we can compute a list of all clusters in~$G[\sA]$ and the size
  of each cluster. The same holds for $G[A']$. Observe that we can always check in~$\Oh(1)$ time, for a given
  vertex~$v$, whether $v$ is contained in~$A'$,~$\sA$,~$\pA$,~$\sB$, or~$\pB$ and, in case~$v$ is contained
  in~$A'$ or~$\sA$, we can find the index and the size of the cluster that contains~$v$.
  Moreover, by Lemma~\ref{lem:few-up} and~\ref{lem:few-down}, we can assume that~$|\pA|=\Oh(k)$, and
  by Lemma~\ref{lem:few-down}, we can assume that~$|\pB|= \Oh(k)$.

  To test the applicability of~\rrrefm{rr:per-verts-m} and~\ref{rr:half-per-edge}, we
  check whether $G[\pA]$ is a cluster graph with at most $k$ clusters, whether~$G[\pB]$ is edgeless, and whether there is an edge with one endpoint
  in~$\pB$ and the other endpoint in~$\sB$. This can be done in~$\Oh(k\cdot
  \Delta)$ time since~$|\pA|=\Oh(k)$,~$|\pB|=\Oh(k)$, and the maximum degree is~$\Delta$.

  To test the applicability of~\rrref{rr:2-per-p3}, we consider each
  pair~$v,w$ of vertices in~$\pA$. If~$v$ and~$w$ are adjacent, then
  in~$\Oh(\Delta)$ time we can check whether there is a vertex~$u\in \sA$
  such that~$u$ is adjacent to exactly one of~$v$ and~$w$. If~$v$
  and~$w$ are not adjacent, then in~$\Oh(\Delta)$ time we can check whether
  they have a common neighbor in~$\sA$. If neither condition applies
  to any pair~$v,w$, then~\rrref{rr:2-per-p3} does not apply. Overall,
  this test takes~$\Oh(k^2\cdot \Delta)$ time.

  To test the applicability of~\brref{br:1-per-p3}, we can check for each
  vertex $v$ of the at most~$k$ vertices of~$\pA$ in~$\Oh(\Delta)$ time whether
  $v$ has neighbors in two different clusters of~$\sA$, or whether there are two vertices $u, w$ in the same cluster of $\sA$ such that $v$ is adjacent to $u$ but not adjacent to~$w$. If one of the two
  cases applies to some vertex~$v\in \pA$, then~\brref{br:1-per-p3}
  applies to~$v$. Otherwise, there is no~$P_3$ containing exactly one
  vertex from~$\pA$ and exactly two vertices from~$\sA$, and~\brref{br:1-per-p3} does not apply. Hence, the applicability of~\brref{br:1-per-p3} can be tested in $\Oh(k \cdot \Delta)$ time.

  To test the applicability of~\brref{br:singleton-cluster}, we
  can check in~$\Oh(k)$ time, whether~$G[\sA]$ contains a singleton cluster
  that is also a singleton cluster of~$G[A']$.

  All rules can trivially be applied in $\Oh(1)$ time if they were found to be applicable.
Hence, the running time to test and apply any of the rules is ~$\Oh(k^2\cdot \Delta)$.

  Finally, if none of the rules applies, then we can check in~$\Oh(k\cdot
  \Delta)$ time whether the number of clusters in~$G[\sA\cup \pA]$ is
  at most~$k$. Hence, the algorithm runs in $\Oh(2^k\cdot
  k^3\cdot \Delta)$ time in total.  \smartqed
\end{proof}
Given the above theorem, Corollary~\ref{cor:in-rec} immediately implies Theorem~\ref{thm:main:mono}, which we restate below:

\begin{reptheorem}{thm:main:mono}
  \thmtxtmainmono
\end{reptheorem}

\section{Generalizations of the algorithm for \textsc{Monopolar Recognition}}\label{sec:gen}
In this section, we present two general \FPT\ algorithms for a
range of cases of~\textsc{$(\Pi_{A},\Pi_{B})$-Recognition} in which
$\Pi_A$ and $\Pi_B$ are characterized by a finite set of forbidden
induced subgraphs. We achieve this by adapting the algorithm of
Section~\ref{sec:mono}, meaning that the obtained algorithms rely on the inductive recognition
framework. Therefore, the main step is to solve \textsc{Inductive
  $(\Pi_{A},\Pi_{B})$-Recognition}, where we are given a graph~$G$ with a distinguished
vertex~$v$, and a~\abpartition{}~$(A',B')$ of~$G-v$, and we are asked to determine whether~$G$ has
a~\abpartition{}~$(A,B)$. To solve \textsc{Inductive $(\Pi_{A},\Pi_{B})$-Recognition}, we
consider again constraints~$\C$ of the type~$(\sA,\pA,\sB,\pB)$. That is, $\C$ is a four-partition of the vertex set such that $\sA \subseteq A'$ and $\sB \subseteq B'$, and $\pA$ and $\pB$ represent the permanent vertices, which may not be moved between $A$ and $B$ anymore. We start with the two initial
constraints~$(A',\{v\},B',\emptyset)$ and~$(A',\emptyset,B',\{v\})$, and recursively search
for solutions fulfilling one of the two constraints, building a search tree whose nodes correspond to constraints. As in Section~\ref{sec:mono}, a \abpartition{}~$(A, B)$ of~$G$ fulfills a constraint~$\C$  if $\pA \subseteq A$ and $\pB \subseteq B$.

The first step towards designing both algorithms is to generalize several reduction and
branching rules of the algorithm for \textsc{Monopolar Recognition}.  The generalization
of~\rrref{rr:per-verts-m} is as follows.
\begin{rrule}\label{rr:fsg}
  If~$G[\pA]$ does not fulfill~$\Pi_A$ or if~$G[\pB]$ does not
  fulfill~$\Pi_B$, then reject the current constraint.
\end{rrule}
\begin{proof}[Proof of correctness]%
  If~$G[\pA]\notin \Pi_A$, then since~$\Pi_A$ is hereditary, there is no~$A$ such that~$\pA\subseteq A$
  and~$G[A]\in \Pi_A$. Similarly, if~$G[\pB]\notin \Pi_B$, then there is no~$B$ such that~$\pB\subseteq B$ and~$G[B]\in \Pi_B$. For
  any~\abpartition{} $(A,B)$ fulfilling~$\C$, we have, however,~$\pA\subseteq A$
  and~$\pB\subseteq B$. Thus, no such~\abpartition{} exists.
\end{proof}

The other reduction and branching rules are aimed at destroying forbidden induced subgraphs in
the candidate vertex set~$\sA \cup \pA$ for~$A$ and in the candidate vertex set~$\sB \cup \pB$ for~$B$ by moving
vertices between~$\sA$ and~$\sB$. To destroy the forbidden induced subgraphs of $\Pi_A$ in~$\sA \cup \pA$, we used~\rrref{rr:2-per-p3} and~\brref{br:1-per-p3} in Section~\ref{sec:mono}. A generalized variant of these rules is as follows.
\begin{brule}\label{br:fsg-a}
  If there is a vertex set~$\tilde{A}\subseteq \sA \cup \pA$ such that~$G[\tilde{A}]$ is a minimal forbidden induced subgraph of~$\Pi_A$, then for each~$u\in \tilde{A}\setminus \pA$ branch into a branch associated with the constraint~$(\sA\setminus\{u\},\pA,\sB,\pB\cup \{u\})$.
\end{brule}
\begin{proof}[Proof of correctness]%
  Suppose that~$(A,B)$ is a~\abpartition{} of $G$ fulfilling~$\C$. Since~$G[\tilde{A}]$ does not
  fulfill~$\Pi_A$, there is a vertex~$u\in \tilde{A}$ such that~$u\in B$. Moreover,
  since~$(A,B)$ fulfills~$\C$, we have~$\pA\subseteq A$, and hence~$\pA\cap
  B=\emptyset$. Therefore,~$u\notin \pA$. Consequently, in the branch of~\brref{br:fsg-a} which is
  associated with the constraint~$(\sA\setminus\{u\},\pA,\sB,\pB\cup \{u\})$, this
  constraint is fulfilled by~$(A,B)$ since~$u\in B$.
\end{proof}

In the case of \textsc{Monopolar Recognition}, for subgraphs that do not fulfill~$\Pi_B$, it was sufficient to use~\rrref{rr:half-per-edge}. This can be generalized to the following branching rule;
the correctness proof is analogous to that of \brref{br:fsg-a} and omitted.

\begin{brule}\label{br:fsg-b}
  If there is a vertex set~$\tilde{B}\subseteq \sB \cup \pB$ such that $G[\tilde{B}]$ is a minimal forbidden induced subgraph of~$\Pi_B$, then for each~$u\in \tilde{B}\setminus \pB$ branch into a branch associated with the constraint~$(\sA,\pA\cup \{u\},\sB\setminus\{u\},\pB)$.
\end{brule}

In the following two subsections we use the above reduction and branching rules and some specialized rule to give the promised algorithms.

\subsection{Cluster Graphs and Graphs Excluding Large Cliques}
\label{sec:k-cluster-large-clique}
 For monopolar graphs, $\Pi_A$ is the family of cluster graphs with at most~$k$ cliques, and~$\Pi_B$
is the property of being edgeless. We now consider the more general case where $\Pi_B$ excludes some clique and
has a characterization by minimal forbidden induced subgraphs, each of order at most~$r$ (and $\Pi_A$ remains the family of all cluster graphs with at most $k$ cliques).

The algorithm uses~\rrref{rr:fsg},~\rrref{rr:2-per-p3},~\brref{br:1-per-p3}, and~\brref{br:fsg-b}. (It does not use \brref{br:fsg-a} which is used the next subsection.)
Note that applying~\rrref{rr:2-per-p3} and~\brref{br:1-per-p3} is
almost equivalent to applying~\brref{br:fsg-a}. The difference is that
we do not branch on forbidden induced subgraphs that are an edgeless
graph on~$k+1$ vertices. This improves the efficiency of the resulting
algorithm. The following fact is not hard to prove.

\begin{fact}
  Let $\C = (\sA,\pA,\sB,\pB)$ be such that none of
  \rrref{rr:fsg},~\rrref{rr:2-per-p3},~\brref{br:1-per-p3},
  and~\brref{br:fsg-b} apply. Then, $G[\sA\cup \pA]$ is a cluster graph
  and~$G[\sB\cup \pB]$ satisfies~$\Pi_B$.
\end{fact}
\noindent That $G[\sA \cup \pA]$ is a cluster graph can be seen by inapplicability of \rrref{rr:fsg},~\rrref{rr:2-per-p3},~\brref{br:1-per-p3} and the fact that none of the rules puts any vertex into~$\sA$ which is initially a subset of~$A'$ of the \abpartition~$(A', B')$. That $G[\sB \cup \pB]$ satisfies~$\Pi_B$ follows from the inapplicability of \rrref{rr:fsg} and \brref{br:fsg-b}.

To obtain a \abpartition~$(\sA \cup \pA, \sB \cup \pB)$, it remains to ensure that $G[\sA \cup \pA]$ contains at most~$k$ clusters. If~$G[\sA\cup \pA]$ has more
than~$k$ clusters, then some vertex has to be moved from~$\sA$
to~$\pB$. This is done in the following branching rule.

\begin{brule}\label{br:too-many-clusters}
  If~$G[\sA\cup \pA]$ is a cluster graph with more than~$k$ clusters, then let $u\in \sA$ be a
  vertex contained in a cluster of~$G[\sA\cup \pA]$, such that this cluster contains no
  vertices from~$\pA$ (such a cluster must exist by~\rrref{rr:fsg}). Branch into two branches: one associated with the
  constraint~$(\sA\setminus\{u\},\pA\cup \{u\},\sB,\pB)$ and one associated
  with the constraint~$(\sA\setminus \{u\}, \pA, \sB, \pB\cup \{u\})$.
\end{brule}
The rule is trivially correct since~$u$ is either contained in~$A$ or in~$B$ for
any~\abpartition~$(A,B)$ fulfilling~$\C$. Now if none of the rules
applies, then we have found a solution.
\begin{fact}\label{fact:found-sol}
  Let $\C=(\sA,\pA,\sB,\pB)$ be such that~\rrref{rr:fsg},~\rrref{rr:2-per-p3}, \brref{br:1-per-p3}, \brref{br:fsg-b}, and \brref{br:too-many-clusters} do not apply to $\C$.
  Then $(\sA\cup\pA,\sB\cup\pB)$ is a~\abpartition.
\end{fact}
To bound the running time, in particular, the number of applications of the branching and
reduction rules, we make use of the fact that the properties~$\Pi_A$ and $\Pi_B$ are mutually $d$-exclusive for some integer~$d$.
\begin{lemma}\label{lem:mutually-exclusive-monopolar-gen}
  Let~$\Pi_A$ be the set of cluster graphs with at
  most~$k$ cliques, and let~$\Pi_B$ be any hereditary graph property that excludes the complete
  graph on~$s$ vertices as an induced subgraph. Then,~$\Pi_A$ and~$\Pi_B$ are mutually~$((s-1)\cdot k+1)$-exclusive.
\end{lemma}
\begin{proof}
  Let~$G$ be a graph of order at least~$(s-1)\cdot k+1$ that fulfills~$\Pi_A$, that is,~$G$ is
  a cluster graph with at most~$k$ clusters. By the pigeonhole principle, one of these
  clusters has at least~$s$ vertices. Therefore, $G$ contains a clique on~$s$
  vertices. Thus,~$G$ does not fulfill~$\Pi_B$.
\end{proof}
We can now conclude with the complete algorithm and its running-time analysis.
\begin{theorem}\label{thm:generalization}
  Let~$\Pi_A$ be the set of all cluster graphs with at most~$k$ cliques, and let~$\Pi_B$ be any hereditary graph property such that
  \begin{itemize}
  \item $\Pi_B$ can be characterized by forbidden induced subgraphs, each of order at most~$r$; and
  \item for some~$s\le r$, the complete graph on~$s$ vertices is a
    forbidden induced subgraph of~$\Pi_B$.
  \end{itemize}
  Then \textsc{$(\Pi_A,\Pi_B)$-Recognition} can be solved in~$2^{s\cdot k}\cdot (r-1)^{(s-1)\cdot k}\cdot n^{\Oh(1)}$
  time.
\end{theorem}
\begin{proof}
  The algorithm creates the two initial constraints~$(\sA,\{v\},\sB,\emptyset)$
  and $(\sA,\emptyset,\sB,\{v\})$. Then, for each initial constraint, it
  applies \rrref{rr:fsg}, \rrref{rr:2-per-p3}, \brref{br:1-per-p3}, \brref{br:fsg-b},
  and \brref{br:too-many-clusters} exhaustively. If none of these rules applies to the
  current constraint, then the algorithm outputs~$(\sA\cup\pA,\sB\cup\pB)$, which, by
  Fact~\ref{fact:found-sol}, is a~\abpartition{} of~$G$. Thus, to prove the correctness of the algorithm,
  it remains to show that if there is a~\abpartition~$(A,B)$ for~$G$, then the algorithm
  outputs such a partition.

  Suppose that $(A,B)$ is a \abpartition{}. Then $(A,B)$ fulfills one of the initial
  constraints. If a constraint~$\C$ is fulfilled by~$(A,B)$, then~\rrref{rr:fsg} does not
  apply to this constraint, and, by the correctness of the rules, any application
  of~\rrref{rr:2-per-p3},~\brref{br:1-per-p3},~\brref{br:fsg-b},
  or~\brref{br:too-many-clusters} yields at least one constraint that is fulfilled
  by~$(A,B)$. Hence, the initial constraint~$\C$ fulfilling~$(A,B)$ has at least one
  descendant that is fulfilled by~$(A,B)$, and to which none of the reduction and branching
  rules applies. For this constraint, the algorithm outputs a~\abpartition.

  It remains to bound the running time of the algorithm. Since all minimal forbidden
  induced subgraphs of~$\Pi_B$ have at most~$r$ vertices, we can check in~$n^{\Oh(1)}$~time, whether any of the branching and reduction rules applies (note that $r$ is a problem-specific constant). To obtain the running time
  bound, it is thus sufficient to bound the number of created constraints in the search
  tree.

  To this end, we bound the number of applications of the branching and reduction rules
  along any path from an initial constraint to a leaf constraint.   \rrref{rr:fsg} is applied at most once. To bound the
  number of applications of the other rules, we use that,
  by~Lemma~\ref{lem:mutually-exclusive-monopolar-gen},~$\Pi_A$ and~$\Pi_B$ are mutually
  $((s-1)\cdot k+1)$ exclusive. This implies that~\brref{br:fsg-b} is
  applied at most~$(s-1)\cdot k+1$ times: Each application of the rule adds a vertex of~$B'$ to~$\pA$. Since~$G[B']$ fulfills~$\Pi_B$, so does~$G[\pA\cap
  B']$. Thus, if~$|\pA\cap B'|>(s-1)\cdot k$, then~$G[\pA]$ does not fulfill~$\Pi_A$
  and~\rrref{rr:fsg} applies, terminating the current branch.

  \rrref{rr:2-per-p3},~\brref{br:1-per-p3}, and~\brref{br:too-many-clusters} can be
  applied altogether at most~$k+(s-1)\cdot k+2$ times: Each application of any of these
  rules either adds a vertex of~$A'$ to~$\pB$ or increases the number of clusters
  in~$G[\pA]$ by one. Again by~Lemma~\ref{lem:mutually-exclusive-monopolar-gen}, at
  most~$(s-1)\cdot k$ vertices of~$A'$ can be moved from~$\sA$ to~$\pB$, before~\rrref{rr:fsg}
  applies. Similarly, if~$G[\pA]$ has more than~$k$ clusters, then~\rrref{rr:fsg} applies.

  To bound the number of leaf constraints, observe that~\brref{br:fsg-b} branches into at most~$r-1$ new constraints, and~\brref{br:1-per-p3} and~\brref{br:too-many-clusters} branch into $2$ new constraints. Thus, the overall number of leaf constraints is~$\Oh(2^{s\cdot k+2}\cdot (r-1)^{(s-1)\cdot k+1})$. By the bound on the number of constraints on any root-leaf path, we thus have that the overall search tree size is~$\Oh(2^{s\cdot k}\cdot (r-1)^{(s-1)\cdot k}\cdot k^{\Oh(1)})$ which implies the overall running time bound.
\end{proof}

\subsection{Mutually Exclusive Graph Properties with Small Forbidden Subgraphs }
\label{sec:mutual-small-fsbg}
We now relax the demands on the hereditary properties~$\Pi_A$ and~$\Pi_B$ even further:
we demand only that~$\Pi_A$ excludes some fixed edgeless graph, that~$\Pi_B$ excludes some
fixed clique, and that $\Pi_A$ and~$\Pi_B$ are each characterized by forbidden induced subgraphs of constant
size. Recall that, by Proposition~\ref{fact:mutually-exclusive}, such properties are
mutually~$d$-exclusive for some constant~$d$.

\begin{theorem}\label{thm:sparse-dense-rec-const-fsg}
  Let~$\Pi_A$ and~$\Pi_B$ be two hereditary graph properties such that
  \begin{itemize}
  \item $\Pi_{A}$ excludes an edgeless graph of order~$c_A$ and has a characterization by
    forbidden induced subgraphs, each of order at most~$r_A$; and
  \item $\Pi_B$ excludes a complete graph of order~$c_B$ and has a characterization by forbidden induced subgraphs of order at most~$r_B$.
  \end{itemize}
  Then \textsc{$(\Pi_{A},\Pi_{B})$-Recognition} can be solved
  in~$(r_A-1)^{R(c_A,c_B)}\cdot (r_B-1)^{R(c_A,c_B)}\cdot n^{\Oh(1)}$ time.
\end{theorem}
\begin{proof}
  We show that \textsc{Inductive $(\Pi_{A},\Pi_{B})$-Recognition} can be solved in this
  running time. In conjunction with Theorem~\ref{thm:in-rec}, this implies the above Theorem~\ref{thm:sparse-dense-rec-const-fsg}.

  The algorithm performs a search from the two initial constraints of
  Fact~\ref{fact:init}. For each constraint encountered during the search, we check
  if \rrref{rr:fsg} applies. If this is not the case, we check if~\brref{br:fsg-a}
  or~\brref{br:fsg-b} applies. If neither applies, then~$(\sA\cup \pA,\sB\cup \pB)$ is
  a~\abpartition{}. Otherwise, apply the respective rule and continue the search with the
  constraints created by the rule. By the correctness of~\brref{br:fsg-a}, \brref{br:fsg-b}, and Fact~\ref{fact:init}, this algorithm finds a~\abpartition{}
  of~$G$ if it exists.

  It remains to analyze the running time of the algorithm. First, observe that throughout the
  algorithm, we have~$\sA\subseteq A'$ and~$\sB\subseteq B'$. Thus, every vertex
  set~$\tilde{A}$ to which~\brref{br:fsg-a} applies contains at least one vertex
  from~$\pA$ and therefore creates at most~$r_A-1$ new recursive branches. Now, observe
  that each application of \brref{br:fsg-a} increases the number of vertices in~$\pB$ by
  one. Moreover, all vertices in~$\pB$ are from~$A'$. Thus, if~$|\pB|\ge R(c_A,c_B)$, then by Proposition~\ref{fact:mutually-exclusive}, this
  implies that~$G[\pB]\notin \Pi_B$. Thus,~\brref{br:fsg-a} is applied at
  most~$R(c_A,c_B)$ times before~\rrref{rr:fsg} applies. Similarly,~\brref{br:fsg-b} is
  applied at most~$R(c_A,c_B)$ times before~\rrref{rr:fsg} applies and each application
  of~\brref{br:fsg-b} creates at most~$r_B-1$ constraints. Overall, the number of created
  constraints is thus~$\Oh((r_A-1)^{R(c_A,c_B)}\cdot(r_B-1)^{R(c_A,c_B)})$. For each
  constraint, we must check if any of the reduction or branching rules applies, which can be done in~$n^{\Oh(1)}$ time by using the polynomial-time algorithms for checking
  membership of a graph in~$\Pi_A$ and~$\Pi_B$.
\end{proof}
For the special case of recognizing monopolar graphs with at most~$k$ cliques,
Theorem~\ref{thm:sparse-dense-rec-const-fsg} applies: The graphs fulfilling $\Pi_A$ have at most~$k$
clusters, thus the edgeless graph of order~$k+1$ is forbidden, implying~$c_A=k+1$. The
forbidden subgraphs for~$\Pi_A$ are exactly the~$P_3$ and the edgeless graph on~$k+1$
vertices, implying~$r_A=k+1$. For~$\Pi_B$, the only forbidden subgraph is the clique on
two vertices, implying~$c_B=r_B=2$. Altogether, this gives a running time
of~$k^{R(k+1,2)}\cdot n^{O(1)}=k^k\cdot n^{O(1)}$. Hence, our tailored algorithm in
Section~\ref{sec:mono} is substantially more efficient than the generic algorithm.

Another application of Theorem~\ref{thm:sparse-dense-rec-const-fsg}
is to~\textsc{$(\Pi_A,\Pi_B)$-Recognition} for~$\Pi_A$ being the triangle-free graphs
and~$\Pi_B$ being the complete graphs. The running time becomes~$\Oh(2^{R(3,2)}\cdot nm)$ in
this case, yielding an~$O(nm)$-time algorithm with very small constant hidden in the
$\Oh$-Notation.
\section{An FPT~Algorithm for \textsc{2-Subcoloring}}
\label{sec:2-subcolor}
In this section, we give an FPT~algorithm for \textsc{2-Subcoloring} parameterized by the smaller number of clusters in the two parts.
Although the general approach is similar to the approach used for \textsc{Monopolar Recognition}, in that it relies on inductive recognition and the notion of constraints, the algorithm is substantially more complex. In particular, the notion of constraints and the reduction and branching rules are more involved. %

Throughout, given a graph $G$ and a nonnegative integer $k$, we call a 2-subcoloring $(A, B)$ of $G$ \emph{valid} if $G[A]$ has at most $k$ cliques. In the inductive recognition framework, we need a parameterized inductive recognizer for the following problem:

\begin{quote}
  \textsc{Inductive 2-Subcoloring}\\
  \textbf{Input:} A graph $G=(V,E)$, a vertex~$v\in V$, and a valid 2-subcoloring~$(A',B')$ of~$G'=G-v$.\\
  \textbf{Question:} Does $G$ have a valid 2-subcoloring~$(A,B)$?
\end{quote}
Fix an instance of \textsc{Inductive 2-Subcoloring} with a graph $G = (V,E)$, a vertex $v \in V$, and a valid 2-subcoloring $(A',B')$ of $G' = G-v$.
We again apply a search-tree algorithm that starts with initial partitions~$(\sA, \sB)$ of~$V$, derived from $(A', B')$, that are not necessarily 2-subcolorings of~$G$. Then, we try to ``repair'' those partitions by moving vertices between~$\sA$ and $\sB$ to form a valid 2-subcoloring $(A,B)$ of~$G$.
As before, %
each node in the search tree is associated with one constraint.

\begin{definition} \rm
  A \emph{constraint}~${\cal C}=(\ssA_1,\ldots, \ssA_k,\ssB_1, \ldots , \ssB_n,\pA,\pB)$ consists of a partition~$(\ssA_1,\ldots,\ssA_k,\ssB_1,\ldots,\ssB_{n})$ of~$V$ and two vertex sets $\pA \subseteq \sA$ and $\pB \subseteq \sB$, where $\sA = \bigcup_{i=1}^{k} \ssA_i$ and $\sB = \bigcup_{i=1}^{n} \ssB_i$, such that for any $i \not= j$:
\begin{compactitem}
\item $u$ and~$w$ are not adjacent for any~$u\in \ssA_i\setminus \pA$ and~$w\in \ssA_j\setminus \pA$, and
\item $u$ and~$w$ are not adjacent for any~$u\in \ssB_i\setminus \pB$ and~$w\in \ssB_j\setminus \pB$.
\end{compactitem}
  We explicitly allow (some of) the sets of the partition~$(\ssA_1,\ldots,\ssA_k,\ssB_1,\ldots,\ssB_n)$ of~$V$ to be empty. The vertices in~$\pA$ and~$\pB$ are called \emph{permanent}
  vertices of the constraint.
\end{definition}

\noindent
The permanent vertices in~$\pA$ and~$\pB$ in the definition will correspond precisely to those
vertices that have switched sides during the algorithm.
We refer to the sets $\ssA_1,\ldots,\ssA_k$ and~$\ssB_1,\ldots,\ssB_{n}$ as \emph{groups}; during the algorithm,~$G[\sA]$ and~$G[\sB]$ are not necessarily cluster graphs and, thus, we avoid using the term clusters.

\looseness=-1 We now define the notion of a valid 2-subcoloring fulfilling a constraint. Intuitively speaking, a constraint ${\cal C}$ is fulfilled by a bipartition~$(A,B)$ if~$(A,B)$
respects the assignment of the permanent vertices stipulated by ${\cal C}$, and if all vertices that do
not switch sides stay in the bipartition $(A,B)$ in the same groups they belong to in ${\cal C}$. This notion is formalized as follows.

\begin{definition} \label{def:4conditions}
  A constraint ${\cal C}=(\ssA_1,\ldots, \ssA_k,\ssB_1, \ldots , \ssB_n,\pA,\pB)$ is \emph{fulfilled} by a bipartition~$(A,B) $ of~$V$
   if $G[A]$ is a cluster graph with~$k$
    clusters~$A_1, \dots , A_k$ and~$G[B]$ is a cluster graph with~$n$
    clusters~$B_1,\dots , B_n$ (some of the clusters may be empty) such that:%
    \vspace{-.4cm}
    \begin{multicols}{2}
      \begin{compactenum}
      \item \label{cond:group-a}for each~$i\in [k]$, $A_i\cap \sA
        \subseteq \ssA_i$;
      \item \label{cond:group-b}for each~$i\in [n]$, $B_i\cap \sB
        \subseteq \ssB_i$;
      \item \label{cond:pa}$\pA\subseteq A$; and
      \item \label{cond:pb}$\pB\subseteq B$.
      \end{compactenum}
    \end{multicols}
    \vspace{-.4cm}
\end{definition}

\noindent
We now need a set of initial constraints to jumpstart the search-tree algorithm.

\begin{lemma}\label{lem:init}
  Let~$A'_1, \ldots ,A'_k$ denote the clusters
  of~$G'[A']$ and let~$B'_1, \ldots ,B'_n$ denote the clusters
  of~$G'[B']$. Herein, if there are less than~$k$ clusters in~$G[A']$ or less than~$n$ clusters in $G[B']$, we add an appropriate number of empty sets. By relabeling, we may assume that only~$B'_1, \ldots ,B'_i$ contain neighbors
  of~$v$, and~$B'_{i+1}=\emptyset$. Each valid 2-subcoloring~$(A,B)$
  of~$G$ fulfills either:
  \begin{compactitem}
  \item $(A'_1,\ldots, A'_j\cup \{v\}, \ldots, A'_k, B'_1, \ldots, B'_n, \{v\}, \emptyset)$ for some~$j\in [k]$, or
  \item $(A'_1,\ldots, A'_k, B'_1, \ldots, B'_j\cup \{v\}, \ldots, B'_n, \emptyset,\{v\})$
    for some~$j\in [i+1]$.
  \end{compactitem}
\end{lemma}
\begin{proof}%
  Since $(A',B')$ is a valid $2$-subcoloring of $G' = G-v$ and $v \in \pA \cup \pB$ for each constraint $\C$, the constructed tuples are indeed constraints. Let~$A_1, \dots, A_k$ be the clusters of~$G[A]$ and
  $B_1, \dots, B_n$ be those of~$G[B]$.

  First, assume that~$v\in A$. If, for some~$j\in [k]$, there is a
  vertex~$u\in A'_j\cap N(v)\cap A$, then~$(A,B)$ fulfills the
  constraint~$(A'_1,\ldots, A'_j\cup \{v\}, \ldots, B'_n, \{v\},
  \emptyset)$. This can be seen as follows. Since~$G[A]$ is a cluster graph,
  each~$A_i$ contains vertices of at most one cluster of~$(A'_1,\ldots,
  A'_k)$. Similarly, each~$B_i$ contains vertices of at most one cluster
  of~$(B'_1,\ldots, B'_n)$. Hence, the clusters of~$G[A]$ and~$G[B]$
  can be labeled accordingly to satisfy Conditions~\ref{cond:group-a}
  and~\ref{cond:group-b} of Definition~\ref{def:4conditions}. Moreover, by assumption, ~$\{v\}\subseteq A$ and
  thus Conditions~\ref{cond:pa} and~\ref{cond:pb} of Definition~\ref{def:4conditions} are fulfilled.

  Similarly, if, for all~$j\in [k]$, vertex~$v$ has no neighbors in~$A'_j\cap
  A$, then there is a~$j\in [k]$ such that~$A'_j\cap
  A=\emptyset$. Then,~$(A,B)$ fulfills the
  constraint~$(A'_1,\ldots, A'_j\cup \{v\}, \ldots, B'_n, \{v\},
  \emptyset)$, by the same arguments as above.

  Symmetric arguments apply for the case~$v\in B$.\smartqed
\end{proof}

\noindent\looseness=-1
Now that we have identified the initial constraints, we turn to the search-tree algorithm and its reduction and branching rules. A crucial ingredient to the rules and the analysis of the running time is the following lemma.
A consequence of the lemma is that if the number of initial constraints is too large, then most of them should be rejected immediately.

\begin{lemma}\label{lem:mustA}
Let ${\cal C}=(\ssA_1,\ldots, \ssA_k,\ssB_1, \ldots , \ssB_n,\pA,\pB)$ be a constraint and let $(A,B)$ be any valid 2-subcoloring of $G$ fulfilling $\C$. If $u \in V$ has neighbors in more than $k+1$ groups among $\ssB_1,\ldots,\ssB_n$, then $u \in A$.
\end{lemma}
\begin{proof}
For the sake of contradiction, suppose that there is a valid 2-subcoloring $(A,B)$ of $G$ fulfilling $\C$ such that $u \in B$. Let $A_1,\ldots,A_k$ and $B_1,\ldots,B_n$ denote the (possibly empty) clusters of $G[A]$ and $G[B]$ respectively. Let $b_1,\ldots,b_{k+2}$ be neighbors of $u$ in distinct groups among $\ssB_1,\ldots,\ssB_n$; these vertices exist by assumption. Without loss of generality, $b_i \in \ssB_i$ for each $i=1,\ldots,k+2$.

Since $(A,B)$ fulfills $\C$, we have $B_i \cap \sB \subseteq \ssB_i$ for each $i=1,\ldots,k+2$. Therefore, the vertices among $b_1,\ldots,b_{k+2}$ that are in $B$ are in distinct clusters of $G[B]$. Since $u \in B$, at least $k+1$ vertices among $b_1,\ldots,b_{k+2}$ are in $A$, say $b_1,\ldots,b_{k+1} \in A$. Since $(A,B)$ fulfills $\C$, we have $b_1,\ldots,b_{k+1} \not\in \pB$, and, thus, $b_i \in \ssB_i \setminus \pB$ for each $i=1,\ldots,k+1$. Then the definition of a constraint implies that $b_1,\ldots,b_{k+1}$ are pairwise nonadjacent. Hence, these vertices are in distinct clusters of $G[A]$. Therefore, $G[A]$ has at least $k+1$ clusters, a contradiction.
\smartqed\end{proof}

Lemma~\ref{lem:mustA} implies that if $v$ has neighbors in more than $k+1$ clusters of $B'$, then we should immediately reject the initial constraints generated by Lemma~\ref{lem:init} that place $v$ in $\pB$. Hence, we obtain the following corollary of Lemmas~\ref{lem:init} and~\ref{lem:mustA}.

\begin{corollary}\label{cor:v-init}
Lemma~\ref{lem:init} generates at most $2k+2$ constraints that are not immediately rejected.
\end{corollary}
As before, each nonroot node of the search tree is associated with a constraint. The root of the search tree is a dummy node with children associated with the constraints generated by Lemma~\ref{lem:init} that are not immediately rejected due to Lemma~\ref{lem:mustA}. We now give two reduction rules, which are applied exhaustively to each search-tree node, in the order they are presented.
Let ${\cal C}=(\ssA_1,\ldots, \ssA_k,\ssB_1, \ldots , \ssB_n,\pA,\pB)$ be a constraint.
The first reduction rule identifies some obvious cases in which the constraint cannot be fulfilled.

\begin{rrule}\label{rr:per-verts}
  If~$G[\pA]$ or~$G[\pB]$ is not a cluster graph, or if there
  are~$i\neq j$ such that there is an edge between~$\ssA_i\cap\pA$
  and~$\ssA_j\cap\pA$ or an edge between~$\ssB_i\cap\pB$ and~$\ssB_j\cap\pB$,
  then reject $\C$.
\end{rrule}
\begin{proof}[Proof of correctness]%
If~$G[\pA]$ is not a cluster graph, then there is no
  valid 2-subcoloring~$(A,B)$ such that~$\pA\subseteq
  A$. Similarly, there is no valid 2-subcoloring $(A,B)$ such
  that~$\pB\subseteq B$ if~$G[\pB]$ is not a cluster graph.

  If there is an edge between two vertices~$u\in \ssA_i\cap\pA$ and~$w\in
  \ssA_j\cap\pA$ where~$i\neq j$, then every valid 2-subcoloring~$(A,B)$
  needs to have~$u$ and~$w$ in the same cluster of~$G[A]$, which
  contradicts Condition~\ref{cond:group-a}. Similarly,
  Condition~\ref{cond:group-b} is violated if there is an edge between two
  vertices~$u\in \ssB_i\cap\pB$ and~$w\in \ssB_j\cap\pB$ where~$i\neq
  j$. \smartqed
\end{proof}
The second reduction rule is the natural consequence of Lemma~\ref{lem:mustA}.

\begin{rrule}\label{rr:make-per-vert}
  If there is a vertex~$u\in \ssA_i\setminus \pA$ that has
  neighbors in more than~$k+1$ groups of~$\sB$, then set~$\pA\leftarrow \pA\cup \{u\}$.
\end{rrule}

The algorithm contains a single branching rule. This rule, called $\switch(u)$, uses branching to fix a vertex $u$ in one of the clusters in one of the parts of the 2-subcoloring. The vertices to which $\switch()$ must be applied are identified by \emph{switching rules}. We say that a switching rule
that calls for applying~$\switch(u)$ is \emph{correct}
if for all valid 2-subcolorings~$(A,B)$ of $G$ fulfilling~$\C$, we have $u \in \sA \cap B$ or $u \in \sB \cap A$.
 We first describe the switching rules, and then describe $\switch(u)$.

The first switching rule identifies vertices that are not adjacent to some
permanent vertices of their group. Recall from Fact~\ref{fact:p3} that cluster graphs do not contain induced $P_3$'s.

\begin{srule}\label{sr:nonedge-within}
  If there is a vertex~$u$ such that~$u\in \ssA_i \setminus \pA$ and~$u$ is not
  adjacent to some vertex in~$\ssA_i\cap \pA$, or~$u\in \ssB_i \setminus \pB$ and~$u$ is
  not adjacent to some vertex in~$\ssB_i\cap \pB$, then
  call~$\switch(u)$.
\end{srule}
\begin{proof}[Proof of correctness]%
  Let~$(A,B)$ be a valid 2-subcoloring fulfilling~$\C$ and let $(A_1,\ldots,A_k)$ be the partition of $A$ induced by the clusters of $G[A]$. We show
  the correctness of the case~$u\in \ssA_i \setminus \pA$ (the other case is
  symmetric). Suppose that $u \in A$. By Condition~\ref{cond:group-a}, $w \in A_i$ for any $w \in \ssA_i \cap A$. Hence, $u \in A_i$. Moreover, since $\pA \subseteq A$ by Condition~\ref{cond:pa}, $\ssA_i \cap \pA \subseteq \ssA_i \cap A$ and, thus, $w \in A_i$ for any $w \in \ssA_i \cap \pA$. However, $G[\{u\}\cup (\ssA_i\cap \pA)]$ is not a clique by assumption, contradicting that $A_i$ is a clique.
\smartqed
\end{proof}
The second switching rule finds vertices that have permanent neighbors
in another group.

\begin{srule}\label{sr:edge-between}
  If there is a vertex~$u$ such that~$u\in \ssA_i \setminus \pA$ and~$u$ has a neighbor
  in~$\pA\setminus \ssA_i$, or~$u\in \ssB_i \setminus \pB$ and~$u$ has a neighbor
  in~$\pB\setminus \ssB_i$, then call~$\switch(u)$.
\end{srule}
\begin{proof}[Proof of correctness]%
  Let~$(A,B)$ be a valid 2-subcoloring fulfilling~$\C$. We show the
  correctness of the case~$u\in \ssA_i \setminus \pA$ (the other case is
  symmetric). If~$u\in A$, then $u$ must be in the same cluster of~$G[A]$
  as its neighbor in~$\pA\setminus \ssA_i$, because $\pA \subseteq A$ by Condition~\ref{cond:pa}. However, this contradicts
  Condition~\ref{cond:group-a}. \smartqed
\end{proof}
Now, we describe $\switch(u)$, which is a combination of a reduction rule and a branching rule. There are two main scenarios that we distinguish. If~$u$ has permanent
neighbors in the other part, then there is only one choice for
assigning~$u$ to a group. Otherwise, we branch into all (up to
symmetry when a group is empty) possibilities to place~$u$ into a group. It is important to note that the switching rules never apply $\switch(u)$ to a permanent vertex.

\begin{brule}[$\switch(u)$] \label{br:switch} \
\vspace{-0.2cm}  \begin{itemize}
  \item If~$u\in \ssA_i \setminus \pA$ and~$u$ has a permanent neighbor in some~$\ssB_j$,
    then set~$\ssA_i\leftarrow \ssA_i\setminus\{u\}$,~$\ssB_j\leftarrow \ssB_j\cup
    \{u\}$,~$\pB\leftarrow \pB \cup\{u\}$.
  \item If~$u\in \ssA_i \setminus \pA$ and~$u$ has only nonpermanent neighbors in~$\sB$,
    then, for each~$\ssB_j$ such that~$N(u)\cap \ssB_j\neq \emptyset$ and~$\ssB_j\cap \pB=\emptyset$, and for one~$\ssB_j$ such that~$\ssB_j=\emptyset$ (chosen arbitrarily), branch into a branch associated with the
    constraint~$(\ssA_1,\ldots, \ssA_i\setminus \{u\}, \ldots, \ssA_k, \ssB_1, \ldots, \ssB_j\cup
    \{u\}, \ldots, \ssB_n, \pA, \pB\cup \{u\})$.
  \item If~$u\in \ssB_i \setminus \pB$ and~$u$ has a permanent neighbor in some~$\ssA_j$,
    then set~$\ssB_i\leftarrow \ssB_i\setminus\{u\}$,~$\ssA_j\leftarrow \ssA_j\cup
    \{u\}$,~$\pA\leftarrow \pA \cup\{u\}$.
  \item If~$u\in \ssB_i \setminus \pB$ and~$u$ has only nonpermanent neighbors in~$\sA$,
    then for each~$\ssA_j$ with~$\ssA_j\cap \pA=\emptyset$, branch into a branch associated with the
    constraint~$(\ssA_1,\ldots, \ssA_j\cup \{u\}, \ldots, \ssA_k, \ssB_1, \ldots, \ssB_i\setminus
    \{u\}, \dots, \ssB_n, \pA\cup \{u\}, \pB)$; if no such $\ssA_j$ exists, reject $\C$.
  \end{itemize}
\end{brule}
\begin{proof}[Proof of correctness]%
  All tuples produced by \brref{br:switch} are indeed constraints, since~$u$ is made permanent in every case. Let~$(A,B)$ be a valid
  2-subcoloring fulfilling~$\C$. We prove that $(A, B)$ fulfills one of the constraints created by \brref{br:switch}. We consider two cases.

  \emph{Case 1: $u\in \ssA_i\setminus \pA$.}  Since the switching rule
  calling~$\switch(u)$ is correct, we have~$u\in B$ and, thus, $\pB\cup
  \{u\}\subseteq B$. Thus, for all constraints
  produced by~$\switch(u)$, the 2-subcoloring~$(A,B)$ satisfies Conditions~\ref{cond:pa} and~\ref{cond:pb} of Definition~\ref{def:4conditions}. Furthermore, Condition~\ref{cond:group-a} is satisfied in every generated constraint because removing $u$ from~$\sA$ weakens the requirement placed on~$A$. It remains to show that $(A, B)$ satisfies Condition~\ref{cond:group-b} for one of the generated constraints. If~$u$ has a neighbor~$w\in \pB$, then~$u$ is in the same cluster as~$w$ in~$G[B]$ since $(A, B)$ fulfills~$\C$ and by Condition~\ref{cond:group-b} of fulfilling constraints. Hence, for the (single) constraint produced by~$\switch(u)$ in this situation, the 2-subcoloring~$(A, B)$ satisfies Condition~\ref{cond:group-b}.
  Now consider the case that~$u$ has no neighbor in~$\pB$. Let~%
  $B_1, \ldots, B_n$ be the clusters in~%
  $G[B]$, %
  indexed according to the groups in~$\C$ and padded with empty sets if the number of clusters is smaller than~$n$. Let $B_j$ be the cluster that contains~$u$. As $u$ does not have neighbors in~$\pB$, cluster~$B_j$ does not contain any vertex of~$\pB$. There are two subcases. First, $u$ has a neighbor~$w \in \sB$. As $(A, B)$ fulfills $\C$ and by Condition~\ref{cond:group-b}, we have~$w \in \ssB_j$. Thus, $(A, B)$ satisfies Condition~\ref{cond:group-b} for the constraint generated by~$\switch(u)$ in which $u$~is added to~$\ssB_j$. Second, $u$ does not have a neighbor in~$\sB$. Hence, $B_j \cap \sB = \emptyset$ and hence, without loss of generality, by relabeling we have $\ssB_j = \emptyset$. Therefore, for the constraint in which $u$ is added to~$\ssB_j = \emptyset$, the 2-subcoloring~$(A, B)$ satisfies Condition~\ref{cond:group-b}.

  \emph{Case 2: $u\in \ssB_i\setminus \pB$.} In this case the argument is
  simpler. Since the switching rule invoking~$\switch(u)$ is
  correct, we have~$u\in A$, and thus $\pA\cup \{u\}\subseteq A$. This
  means that, for all constraints
  produced by~$\switch(u)$, the 2-subcoloring~$(A,B)$ fulfills Conditions~\ref{cond:pa}
  and~\ref{cond:pb}. The case that~$u$ has
  a neighbor in~$w\in \pA$ follows by an argument symmetric to the
  one of Case~1.
  If~$u$ has no neighbor~$w\in \pA$,
  then~$\switch(u)$ considers all possibilities of placing~$u$ in one of
  the clusters of~$A$. Again, $u$ cannot be in a cluster in $G[A]$ that contains a vertex from~$\pA$, meaning that the groups containing permanent
  vertices may be ignored in the branching. Hence, for one of the produced constraints, $(A,B)$
  fulfills~Conditions~\ref{cond:group-a} and~\ref{cond:group-b}.
  However, special consideration is needed if
  each of $\ssA_1,\ldots,\ssA_k$ has a permanent vertex. Let $(A,B)$ be a valid 2-subcoloring of~$G$ that fulfills~$\C$. Then Condition~\ref{cond:group-a} and~\ref{cond:pa} imply that $G[A]$ has $k$ clusters that each contains a permanent vertex of $\pA$. However, $u$ is not adjacent to any permanent vertices. Hence, $u \not\in A$, which contradicts the correctness of the switching rule that called $\switch(u)$. Therefore, $\C$ cannot be fulfilled, and the rule correctly rejects~$\C$.
\smartqed
\end{proof}

If none of the previous rules applies, then the
constraint directly gives a solution:

\begin{lemma}\label{lem:init-constraint-sc}
  Let~${\cal C}=(\ssA_1,\ldots, \ssB_n,\pA,\pB)$ be a constraint such that
  none of the rules applies. Then~$(\sA,\sB)$ is a valid
  2-subcoloring.
\end{lemma}
\begin{proof}
  We need to show that~$G[\sA]$ and~$G[\sB]$ are cluster graphs
  and that~$G[\sA]$ has at most~$k$ clusters. First, we claim that~$G[\ssA_i]$ is a clique for every~$i = 1,\ldots,k$. Every vertex
  in~$\ssA_i \setminus \pA$ is adjacent to every vertex in~$\ssA_i\cap \pA$;
  otherwise,~\srref{sr:nonedge-within} applies. Any two vertices
  in~$\ssA_i\setminus \pA$ are also adjacent, because they are in the same
  cluster of $A'$. It remains to show that $G[\ssA_i \cap \pA]$ is a clique. By the description of $\switch(u)$, if a vertex $x$ is placed into $\ssA_i$ and $\ssA_i \cap \pA \not= \emptyset$, then $x$ is adjacent to a vertex of $\ssA_i \cap \pA$. Hence, $G[\ssA_i \cap \pA]$ is connected. Since~\rrref{rr:per-verts} does not apply, $G[\ssA_i \cap \pA]$ does not contain an induced $P_3$ and, thus, it is a clique. Hence, $G[\ssA_i]$ is a clique, as claimed.

  Second, we claim that there are no
  edges between~$\ssA_i$ and~$\ssA_j$, where~$i\neq j$. Suppose for the sake of a contradiction that $e$ is such an edge. Since~\rrref{rr:per-verts} does not apply, $e$ is incident with at least one nonpermanent vertex. Since~\srref{sr:edge-between} does not apply, $e$ is in fact incident with two nonpermanent vertices. Then $e$ cannot exist by the definition of a constraint. The claim follows.

\looseness=-1
  The combination of the above claims shows that~$G[\sA]$ is a cluster graph with the clusters $\ssA_i$ (some of which may be empty) and, thus, has at most~$k$
  clusters. Similar arguments show that~$G[\sB]$ is a cluster graph: in the above argument, we used only \rrref{rr:per-verts} and~\srrefm{sr:nonedge-within} and~\ref{sr:edge-between}, which apply to vertices in~$\sA$ and~$\sB$ symmetrically. \smartqed
\end{proof}
Using the above rules and lemmas, we can now show the following.

\begin{theorem}\label{thm:ind2col}
\textsc{Inductive 2-Subcoloring} can be solved in $\Oh(k^{2k + 1}\cdot(n+m))$ time.
\end{theorem}
\begin{proof}\looseness=-1
Given the valid 2-subcoloring $(A',B')$ of $G'$, we use Lemma~\ref{lem:init} to generate a set of initial constraints, and reject those which cannot be fulfilled due to Lemma~\ref{lem:mustA}. By Corollary~\ref{cor:v-init}, at most $2k+2$ initial constraints remain, which are associated with the children of the (dummy) root node. For each node of the search tree, we first exhaustively apply the reduction rules on the associated constraint. Afterwards, if there exists a vertex~$u$ to which a switching rule applies, then we apply $\switch(u)$. If $\switch(u)$ does not branch but instead reduces to a new constraint, then we apply the reduction rules exhaustively again, etc.

A \emph{leaf} of the search tree is a node associated either with a constraint that is rejected, or with a constraint to which no rule applies. The latter is called an \emph{exhausted leaf}. If the search tree has an exhausted leaf, then the algorithm answers `yes'; otherwise, it answers `no'. By the correctness of the reduction, branching, and switching rules, and by Lemma~\ref{lem:init-constraint-sc}, graph~$G$ has a valid 2-subcoloring if and only if the search tree has at least one exhausted leaf node. Therefore, the described search-tree algorithm correctly decides an instance of \textsc{Inductive 2-Subcoloring}.

We now bound the running time of the algorithm. Observe that each described reduction rule and the branching rule $\switch()$ either rejects the constraint or makes a vertex permanent. Hence, along each root-leaf path, $\Oh(n)$ rules are applied. Each rule can trivially be tested for applicability and applied in polynomial time. Hence, it remains to bound the number of leaves of the search tree.

As mentioned, at the root of the search tree, we create at most~$\Oh(n)$ constraints, out of which at most $2k+2$ constraints do not correspond to leaf nodes by Lemma~\ref{lem:init}, Corollary~\ref{cor:v-init} and \rrref{rr:make-per-vert}. The only branches are created by a call to $\switch(u)$ for a vertex~$u$ that has only nonpermanent neighbors in the
  other part of the bipartition~$(\sA,\sB)$. Observe that
  if such a vertex~$u\in \sB \setminus \pB$, then in each constraint~$\C'$ constructed by
  $\switch(u)$ the number of groups in~$A^{\C'}_{*}$ that have at least
  one permanent vertex increases by one compared to $\C$. Since each constraint has~$k$
  groups in~$\sA$, this branch can be applied at most~$k$ times along each root-leaf path in the search tree.

\looseness=-1
  Similarly, if~$u\in \sA \setminus \pA$, then in each constraint~$\C'$ constructed by
  $\switch(u)$ the
  number of groups in~$\sB$ that have at least one permanent vertex
  increases by one compared to $\C$. We claim that, if~$\sB$ has~$k$ groups with a
  permanent vertex, then~$u$ has a neighbor
  in~$\pB$. First, each permanent vertex in~$\sB$ is part
  of~$A'$ by the description of the rules. Moreover, the permanent vertices of the $k$ groups in~$\sB$ with a permanent vertex stem from $k$ different
  clusters in~$G[A']$, because $\switch()$ places a vertex of $\sA \setminus \pA$ that has neighbors in $\pB$ in the same group as its neighbors in $\pB$. This implies that one of the clusters in~$G[A']$ that the permanent vertices stem from contains~$u$. Hence, $u$ is adjacent to a vertex in $\pB$, as claimed.
  The claim implies that if $\sB$ has $k$ groups with a permanent vertex, then $\switch(u)$ applied to a vertex $u \in \sA \setminus \pA$ does not branch. Hence, also the branch of $\switch(u)$ in which $u \in \sA \setminus \pA$ is performed at most~$k$ times along each root-leaf path in
  the search tree.

  In summary, the branchings of $\switch(u)$ in which $u\in \sB \setminus \pB$ branch into at most~$k$ cases, and
  the branchings in which $u \in \sA \setminus \pA$ branch into at most~$k+2$ cases,
  since~\rrref{rr:make-per-vert} does not apply. Observe
  that~$k$ of the initial constraints have already one group in~$\sA$ with
  a permanent vertex, and the other~$k+1$ initial constraints have one
  group in~$B$ with a permanent vertex. Thus, if the initial
  constraint~$\C$ places~$v$ in~$\pA$, then the overall number of
  constraints from~$\C$ by branching is at most
  $k^{k-1}\cdot (k+2)^k.$  If the initial constraint~$\C$ places~$v$ in~$\pB$,
  then the overall number of constraints created from~$\C$ by
  branching is at most
  $k^{k}\cdot (k+2)^{k-1}.$
  Altogether, the number of constraints created by branching is thus
  $$(2k+1)\cdot
  k^k\cdot (k+2)^k=(2k+1) \cdot k^k \cdot k^k \cdot [(1+1/(k/2))^{k/2}]^2 =\Oh(k^{2k+1})$$
   after noting that  $[(1+1/(k/2))^{k/2}]^2 =\Oh(1)$. This provides the claimed bound on the number of leaves of the search tree. We now give the detailed proof of the running time.

  Our goal is to achieve (almost)
  linear running time per search-tree node. To this end, we pursue the
  following strategy for applying the reduction, switching, and
  branching rules. Note that each rule except \rrref{rr:make-per-vert}
  can become applicable only due to making a vertex permanent, that
  is, by placing a vertex into $\pA$ or $\pB$. Hence, it suffices to
  check whether any rule applies whenever we make a vertex permanent;
  the applications of \rrref{rr:make-per-vert} which are not caused by
  making a vertex permanent receive special treatment below. We now
  argue that, whenever we make a vertex~$u$ permanent, we can
  determine in $\Oh(k \cdot \deg(u))$~time, whether any rule applies
  (and a vertex to which it applies), after an initial, one-time
  expense of $\Oh(m + n)$~time. After this, we prove that this is
  enough to show the running time bound of $\Oh(k^{2k + 1}\cdot(m +
  n))$.

  First, observe that we can initialize and maintain in $\Oh(m + n)$
  time throughout the algorithm a data structure that allows us to
  determine in $\Oh(1)$ time for an arbitrary vertex to which group it
  belongs and whether it is permanent, and to determine in $\Oh(1)$
  time for an arbitrary group how many permanent vertices it contains.
  We additionally maintain a data structure that allows us to determine in $\Oh(1)$
  time for an arbitrary permanent vertex to which cluster it belongs
  in $G[\pA]$ or $G[\pB]$, and to determine in $\Oh(1)$ time for an
  arbitrary cluster in $G[\pA]$ or $G[\pB]$ how many vertices it
  contains. This can also be done in $\Oh(m + n)$ time overall, since
  each vertex becomes permanent only once.

  To check whether
  \rrref{rr:per-verts} becomes applicable when we make a vertex~$u$
  permanent, it suffices to check whether all permanent
  neighbors of~$u$ are in the same cluster and whether these
  neighbors include all vertices in this cluster. This can clearly be
  done in $\Oh(\deg(u))$ time using the data structures mentioned above.

  For \srref{sr:nonedge-within}, we first iterate over the $\deg(u)$ neighbors of~$u$, labeling each with a ``timestamp'', an integer that is initially~0 and increases whenever we make a vertex permanent. Then, we iterate over a list of vertices in the group of~$u$ (note that, within overall linear
  time per search tree node, we can maintain these lists for all
   groups). For each vertex~$v$ in this list, we check in $\Oh(1)$~time  whether it is labeled with the current timestamp. If not, then~$u$ and~$v$ are not adjacent and \srref{sr:nonedge-within} applies. After iterating over at most $\deg(u) + 1$ vertices, we
  encounter a vertex that is nonadjacent to $u$ if there is one. %

  We can check for the applicability of \srref{sr:edge-between} in
  $\Oh(\deg(u))$ time by examining each neighbor of $u$ using the
  aforementioned data structures.

  For \rrref{rr:make-per-vert}, note first that, except for the possible
  applications after the initial constraints have been created,
  \rrref{rr:make-per-vert} can only become applicable when we move a
  vertex to~$\sB$. Whenever we move a vertex to $\sB$ in any of the
  rules, we also make it permanent. Except for the initial
  applications, which we treat below, it thus suffices to check
  whether \rrref{rr:make-per-vert} becomes applicable whenever we make
  a vertex~$u$ permanent. To do this in $\Oh(k\cdot\deg(u))$~time, we
  maintain throughout the algorithm for each nonpermanent vertex
  in~$\sA$ a list with at most~$k$ entries containing the indices of the groups
  in~$\sB$ in which it has neighbors. This list can be initialized in
  $\Oh(m + n)$
  time in the beginning. Whenever we
  move a vertex~$u$ to $\sA$ (and make it permanent), we initialize
  such a list for $u$. Whenever we move a vertex~$u$ to $\sB$ (and
  make it permanent), we update the list for each neighbor. Both
  cases take $\Oh(k \cdot \deg(u))$~time. At the same time, we
  can check whether \rrref{rr:make-per-vert} becomes applicable and
  identify the vertex to which it becomes applicable.

  Concluding, whenever we make a vertex $u$ permanent, we can
  determine in $\Oh(k\cdot\deg(u))$ time whether any rule becomes
  applicable, and, if so, find a corresponding vertex to apply it to if
  needed. We now show that this suffices to prove an overall running
  time of $\Oh(k^{2k + 1}\cdot(m + n))$~time.

  To see this, note first that each reduction and switching rule,
  including the reduction part of \brref{br:switch}, can be carried
  out in $\Oh(1)$~time, once we have determined whether they are
  applicable, and a vertex to which they are applicable if they need
  one. (By carrying out a switching rule, we mean to add the
  vertex~$u$ to a queue of vertices to which we shall
  apply~$\switch()$.)  Furthermore, we make each vertex permanent only
  once. Hence, the time for checking the applicability and for
  applying the reduction and switching rules along each root-leaf path
  in the search tree is bounded by $\sum_{u \in V}\Oh(k\cdot\deg(u)) =
  \Oh(k\cdot(m + n))$. Herein, we include the reduction part of
  \brref{br:switch} but exclude the applications of
  \rrref{rr:make-per-vert} that are not due to making a vertex
  permanent. This running time also subsumes the time for the branching
  part of \brref{br:switch} by attributing the time
  taken for constructing constraints in a search-tree node to its child nodes.

  To finish the proof, it remains to treat \rrref{rr:make-per-vert} and
  the pruning of the initial constraints generated by
  Lemma~\ref{lem:mustA}. For \rrref{rr:make-per-vert}, observe that,
  after an initial exhaustive application when we place~$v$,
  the only applications are due to moving a vertex between~$\sA$
  and~$\sB$, which we already accounted for above. Furthermore,
  \rrref{rr:make-per-vert} can be exhaustively applied in~$\Oh(m + n)$
  time, since no vertices are moved to $\sB$ in the process. Hence,
  the overall time needed for \rrref{rr:make-per-vert} along a root-leaf path in the search tree is
  $\Oh(k\cdot(m + n))$. For the initial constraints generated by
  Lemma~\ref{lem:mustA}, observe that we can check in linear time
  whether $v$ has more than $k$ neighbors in groups in $\sB$, and, if
  so, make $v$ permanent in $\sA$; then we only generate the appropriate
  at most~$k$ constraints. Otherwise, we can directly create the at
  most $2k + 2$ constraints. In both cases, the initial constraints can be created in~$\Oh(k\cdot (m+n))$ time.  
  The overall time upper bound of $\Oh(k^{2k + 1}\cdot(m + n))$ follows.
   \smartqed
\end{proof}
Given the above theorem, we obtain our running time bound for~\textsc{2-Subcoloring}.

\begin{reptheorem}{thm:main:sub}
  \thmtxtmainsub
\end{reptheorem}
\begin{proof}
Theorem~\ref{thm:ind2col}  and Corollary~\ref{cor:in-rec-k} immediately imply a running time bound of $\Oh(k^{2k+1}\cdot (n^2+nm))$. Moreover, before starting inductive recognition, we may remove all vertices of degree 0 because they can be safely added to~$B$. Afterwards,~$n=\Oh(m)$, giving the claimed running time.
\end{proof}

\section{Further FPT Results}
\label{sec:totalpolar}
In this section, we consider two examples of parameterized \textsc{$(\Pi_A,\Pi_B)$-Recognition} problems for which the branching technique seems to outperform inductive recognition, either in terms of simplicity or efficiency. The first example we consider is that of $(\Pi_A,\Pi_B)$-Recognition problems, in which $|A| \leq k$, where $k$ is the parameter, and $\Pi_B$ can be characterized by a finite set of forbidden induced subgraphs.
We show that these problems can be solved in time $2^{\Oh(k)}n^{\Oh(1)}$ by a straightforward branching strategy. We note that the inductive recognition technique can be used to solve these problem within the same running time, albeit with a little more work. The second example we consider is {\sc 2-Subcoloring} parameterized by the total number of clusters in both sides of the partition. We show that this problem
can be solved in $\Oh(4^k \cdot k^2 \cdot n^2)$ time by a branching strategy, followed by reducing the resulting instances to the \textsc{2-CNF-Sat} problem.
We note that the algorithm in the previous section for {\sc 2-Subcoloring} can be modified to solve this problem, but runs in time $\Oh(k^k nm)$.
\subsection{Parameterizing $(\Pi_A,\Pi_B)$-Recognition by the Cardinality of $A$}
\label{subsec:parambyA}
In this subsection, we consider the parameter consisting of the total number of vertices in $G[A]$. The following proposition shows that a straightforward branching strategy yields a generic fixed-parameter algorithm for many \textsc{$(\Pi_A,\Pi_B)$-Recognition} problems:

\begin{repproposition}{prp:main:size}
  \prptxtmainsize
\end{repproposition}
\begin{proof}
We describe a branching algorithm. Initially, let $A = \emptyset$ and $B = V$. Now consider a branch where we are given $(A,B)$. If $|A| > k$ or $G[A] \notin \Pi_A$, then reject the current branch. Using the assumed finite set $\mathcal{H}$ of forbidden induced subgraphs that characterizes $\Pi_B$, find a forbidden induced subgraph $H=(W,F)$ in $G[B]$. This takes polynomial time, since $\mathcal{H}$ is finite. If $H$ does not exist, then accept the current branch and answer `yes'. Otherwise, branch into $|W|$ branches: for each $w \in W$, we construct the branch where $A' = A \cup \{w\}$ and $B' = B \setminus\{w\}$. If the algorithm never accepts a branch, then we answer `no'. The correctness of the algorithm is straightforward: for each forbidden induced subgraph that the algorithm finds, there is a branch where the algorithm adds the vertex to $A$ that is also in $A$ in a solution. Every node in the search tree has finite degree, since $|W|$ is finite. Moreover, the search tree has depth at most $k$, since each branch adds a single vertex to $A$ and $|A|$ cannot exceed $k$. Hence, the running time is~$2^{\Oh(k)}n^{\Oh(1)}$, as claimed.
\smartqed\end{proof}
If~$\Pi_B$ has only an infinite characterization by forbidden subgraphs, then \textsc{$(\Pi_A,\Pi_B)$-Recognition} parameterized by~$|A|$ is W[2]-hard in some cases: The problem of deleting at most~$k$ vertices in an undirected graph such that the resulting graph has no so-called wheel as induced subgraph is W[2]-hard with respect to~$k$~\cite{Lok08}. Thus, if~$\Pi_A$ is the class of graphs of order at most~$k$ and~$\Pi_B$ is the class of wheel-free graphs, then~\textsc{$(\Pi_A,\Pi_B)$-Recognition} is W[2]-hard with respect to~$|A|$.
\subsection{Parameterizing 2-Subcoloring by the Total Number of Clusters}
\label{subsec:totalparam}

\newcommand{\valid}{valid}

In this subsection, we prove Theorem~\ref{thm:main:sub2} by presenting an \FPT \ algorithm for {\sc 2-Subcoloring} parameterized by the total number of clusters in both sides of the partition.
 Throughout, given a graph $G$ and a nonnegative integer $k$, we call a 2-subcoloring $(A,B)$ of $G$ \emph{\valid} if $G[A]$ and $G[B]$ are cluster graphs that, in total, have at most $k$ clusters. %

 As mentioned, the presented algorithm does not rely on inductive recognition, but instead performs branching on the entire graph, followed (possibly) by reducing the resulting instances to the \textsc{2-CNF-Sat} problem. To describe the branching algorithm, we define a notion of a constraint. Throughout, let $k$ be a nonnegative integer and let $G = (V,E)$ be a graph for which we want to decide whether it has a \valid\ 2-subcoloring.

\begin{definition}
Let $k_1,k_2 \in \mathbb{N}$. A \emph{constraint} $\C$ is a partition \[(\ssA_{1},\ldots,\ssA_{k_{1}}, \ssB_{1},\ldots,\ssB_{k_{2}}, \rC)\] of $V$ such that $G[\sA]$ is a cluster graph with the clusters $\ssA_{1},\ldots,\ssA_{k_{1}}$ and $G[\sB]$ is a cluster graph with the clusters $\ssB_{1},\ldots,\ssB_{k_{2}}$, where $\sA = \bigcup_{i=1}^{k_1} \ssA_i$ and $\sB = \bigcup_{j=1}^{k_2} \ssB_j$. We explicitly allow $k_1 = 0$ or $k_2 = 0$, meaning that there are no parts $\ssA_{i}$ or no parts $\ssB_{i}$ in these cases.
A constraint is \emph{fulfilled} by a \valid\ 2-subcoloring $(A,B)$ of $G$ if $\sA \subseteq A$ and $\sB \subseteq B$.
\end{definition}
The set $\rC$ in a constraint $\C$ denotes the set of remaining vertices that have not been assigned to any cluster (yet).

The following proposition is crucial at several places in the algorithm:

\begin{proposition} \label{prp:total-obs}
If a constraint $\C = (\ssA_{1},\ldots,\ssA_{k_{1}}, \ssB_{1},\ldots,\ssB_{k_{2}}, \rC)$ can be fulfilled by a \valid\ 2-subcoloring $(A,B)$ of $G$, then $G[A]$ contains $k_1$ clusters (say $A_1,\ldots,A_{k_1}$) and $G[B]$ contains $k_2$ clusters (say $B_1,\ldots,B_{k_2}$) such that $\ssA_i \subseteq A_i$ for each $i \in [k_1]$ and $\ssB_j \subseteq B_j$ for each $j\in [k_2]$.
\end{proposition}
\begin{proof}
Recall that $\sA \subseteq A$ and $\sB \subseteq B$. Moreover, observe that $u \in \ssA_i$ for $i\in [k_1]$ and $v \in \ssA_{i'}$ for $i'\in [k_1]$ are adjacent if and only if $i=i'$, because $G[\sA]$ is a cluster graph with the clusters $\ssA_{1},\ldots,\ssA_{k_{1}}$. Hence, $u \in \ssA_i$ for $i \in [k_1]$ and $v \in \ssA_{i'}$ for $i'\in [k_1]$ are in the same cluster of $G[A]$ if and only $i=i'$. A similar observation holds with respect to $B$.
\smartqed\end{proof}

 The algorithm is a search-tree algorithm. Each node in the search tree
is associated with a constraint. The root of the search tree is associated with the constraint $\C_r = (\rC_r = V)$ (that is, all vertices of $G$ are in $\rC_r$ and there are no clusters). At a node in the search tree, the algorithm searches for a solution that fulfills the associated constraint~$\C$. To this end, the algorithm applies reduction and branching rules that successively tighten~$\C$.
We describe these rules next. As usual, reduction rules are performed exhaustively before branching rules, and the rules are performed in the order they are presented below.
Throughout, let $\C = (\ssA_{1},\ldots,\ssA_{k_{1}}, \ssB_{1},\ldots,\ssB_{k_{2}}, \rC)$ be a constraint associated with a node of the search tree of the algorithm, and let $\sA = \bigcup_{i=1}^{k_1} \ssA_i$ and $\sB = \bigcup_{j=1}^{k_2} \ssB_j$.

\begin{rrule} \label{rr:total-1}
If $k_1 + k_2>k$, then reject $\C$.
\end{rrule}
\begin{proof}[Proof of correctness]%
By Proposition~\ref{prp:total-obs}, no $2$-subcoloring of $G$ that fulfills $\C$ can be \valid.
\smartqed\end{proof}
For simplicity, we call $\ssA_{1},\ldots,\ssA_{k_{1}}$ the \emph{clusters of $\sA$} and $\ssB_{1},\ldots,\ssB_{k_{2}}$ the \emph{clusters of $\sB$}.  When we \emph{open} a new cluster in $\sA$ with a vertex $v \in \rC$, we create the constraint $\C' = (A^{\C'}_{1},\ldots,A^{\C'}_{k_{1}+1}, B^{\C'}_{1},\ldots,B^{\C'}_{k_{2}}, R^{\C'})$ where $A^{\C'}_i = \ssA_i$ for $i\in [k_1]$, $A^{\C'}_{k_1+1} = \{v\}$, $B^{\C'}_j = \ssB_j$ for $j\in [k_2]$, and $R^{\C'} = \rC \setminus \{v\}$. Opening a new cluster in $\sB$ is similarly defined. When we \emph{add} a vertex $v \in \rC$ to $\sA$, we create the constraint $\C'$ that differs from $\C$ in that $v \not\in R^{\C'}$ but instead $v \in A^{\C'}_{i}$, where $v$ is adjacent to all vertices of $\ssA_i$. Adding a vertex to $\sB$ is similarly defined.  Note that we can only open a new cluster or add a vertex to a cluster if this indeed yields a constraint; for the vertices to which these operations are applied during the algorithm, this will always be ensured.

The next reduction rule treats vertices which have to be in~$\sA$ or $\sB$ due to connections to the corresponding clusters.
\begin{rrule} \label{rr:total-2}
If $v \in \rC$ is adjacent to two clusters of $\sA$ (resp.\ $\sB$) or else if $v$ is adjacent to some but not all vertices of a cluster of $\sA$ (resp.\ $\sB$), then
\begin{itemize}
\item if $v$ is adjacent to two clusters of $\sB$ (resp.\ $\sA$) or if $v$ is adjacent to some but not all vertices of a cluster of $\sB$ (resp.\ $\sA$), then reject $\C$;
\item if $v$ is not adjacent to any cluster in $\sB$ (resp.\ $\sA$), then open a new cluster in $\sB$ (resp.~$\sA$) with $v$;
\item otherwise, add $v$ to $\sB$ (resp.\ to $\sA$).
\end{itemize}
\end{rrule}
\begin{proof}[Proof of correctness]%
Suppose that $v \in \rC$ is adjacent to two clusters of $\sA$. Then in any \valid\ 2-subcoloring $(A,B)$ of $G$ that fulfills $\C$, $v$ is also adjacent to two clusters of $A$ by Proposition~\ref{prp:total-obs}. Hence, $v \in B$. Similarly, if $v$ is adjacent to some but not all vertices of a cluster of $\sA$, then $v$ is also adjacent to some but not all vertices of a cluster of $A$ by Proposition~\ref{prp:total-obs}, and thus $v\in B$. By a symmetric argument, this immediately shows that the first part of the rule correctly rejects $\C$.

Suppose that $v$ is not adjacent to any cluster in $\sB$. Then in any \valid\ 2-subcoloring $(A,B)$ of $G$ that fulfills $\C$ and that satisfies $v \in B$, $v$ must be in a different cluster of $G[B]$ than the vertices of $\sB$ by Proposition~\ref{prp:total-obs}, because $v$ is not adjacent to any vertex of $\sB$. Hence, the second part of the rule correctly opens a new cluster with $v$.

If none of the previous parts of the rule applied to $v$, then $v$ is adjacent to all vertices of a cluster $\ssB_i$ of $\sB$. Then in any \valid\ 2-subcoloring $(A,B)$ of $G$ that fulfills $\C$ and that satisfies $v \in B$, $v$ must be in the same cluster as the vertices of $\ssB_i$ by Proposition~\ref{prp:total-obs}. Hence, the third part of the rule correctly adds $v$ to $\ssB_i$.
\smartqed\end{proof}
Several types of ambiguous vertices remain. We next describe four branching rules that treat some of these vertices, the status of the remaining ones will be determined via a reduction to
{\sc 2-CNF-Sat}. The first type of vertices have no connections to any
cluster and are treated by the following rule.
\begin{brule} \label{br:total-1}
If $v \in \rC$ is not adjacent to any cluster of $\sA$ or $\sB$, then branch into two branches: in the first, open a new cluster in $\sA$ with $v$; in the second, open a new cluster in $\sB$ with $v$.
\end{brule}
\begin{proof}[Proof of correctness]
Let $(A,B)$ be any \valid\ 2-subcoloring of $G$ that fulfills $\C$. Suppose that $v \in A$. Then $v$ must be in a different cluster of $G[A]$ than the vertices of $\sA$ by Proposition~\ref{prp:total-obs}, because $v$ is not adjacent to any vertex of $\sA$. Hence, $(A,B)$ fulfills the constraint generated in the first branch. Similarly, if $v \in B$, then $(A,B)$ fulfills the constraint generated in the second branch.
\smartqed\end{proof}
After this rule has been applied, we can obtain a useful partition of $\rC$, classifying vertices according to the clusters in $\sA$ or $\sB$ to which they belong, if they are put into $\sA$ or $\sB$, respectively.

\begin{proposition} \label{prp:total-obs2}
If~\rrref{rr:total-2} and~\brref{br:total-1} cannot be applied, then $\rC$ can be partitioned into sets $N_{\ssA_i}$, $N_{\ssB_j}$, and $N_{\ssA_i\ssB_j}$ for $i\in [k_1]$ and $j\in [k_2]$, where:
\begin{itemize}
\item any vertex in $N_{\ssA_i}$ is adjacent to all vertices of cluster $\ssA_i$ and not to any other clusters of $\sA$ nor any clusters of~$\sB$;
\item any vertex in $N_{\ssB_j}$ is adjacent to all vertices of cluster $\ssB_j$ and not to any other clusters of $\sB$ nor any clusters of~$\sA$; and
\item any vertex in $N_{\ssA_i\ssB_j}$ is adjacent to all vertices of clusters $\ssA_i$ and $\ssB_j$ and not to any other clusters of $\sA$ or $\sB$.
\end{itemize}
\end{proposition}
\begin{proof}
By~\brref{br:total-1}, any vertex $v \in \rC$ is adjacent to at least one cluster of $\sA$ or $\sB$. By~\rrref{rr:total-2}, $v$ cannot be adjacent to two clusters of $\sA$ or two clusters of~$\sB$. Furthermore, again by~\rrref{rr:total-2}, if $v$ is adjacent to a cluster $\ssA_i$ of $\sA$ (resp.\ $\ssB_i$ of $\sB$), then $v$ is adjacent to all vertices of that cluster.
\smartqed\end{proof}
This partition of $\rC$ enables further branching rules. The following two branching rules take care of pairs of vertices that cannot belong to the same cluster in one of the parts~$\sA$, $\sB$.
\begin{brule} \label{br:total-2}
If $u,v \in N_{\ssA_i}$ for some $i\in [k_1]$ (resp.\ $u,v\in N_{\ssB_j}$ for some $j\in [k_2]$) and $uv \not\in E$, then branch into two branches: in the first, open a new cluster in $\sB$ (resp.\ $\sA$) with $u$; in the second, open a new cluster in $\sB$ (resp.\ $\sA$) with $v$.
\end{brule}
\begin{proof}[Proof of correctness]
Let $(A,B)$ be any \valid\ 2-subcoloring of $G$ that fulfills $\C$. Suppose $u,v\in N_{\ssA_i}$ for some $i\in [k_1]$. Then $\ssA_i \subseteq A_i$ for some cluster $A_i$ of $G[A]$ by Proposition~\ref{prp:total-obs}. Since $uv \not \in E$, at least one of $u,v$ is not in $A_i$. In fact, since $u$ and $v$ are adjacent to all vertices of $\ssA_i$ by definition, at least one of $u,v$ is not in $A$. By definition, $u$ and $v$ are not adjacent to any vertex of $\sB$. Hence, at least one of $u,v$ must be in a different cluster of $G[B]$ than the vertices of $\sB$ by Proposition~\ref{prp:total-obs}. Therefore, $(A,B)$ fulfills the generated constraint in at least one of the branches. The argument in case $u,v\in N_{\ssB_j}$ for some $j\in [k_2]$ follows symmetrically.
\smartqed\end{proof}
\begin{brule} \label{br:total-3}
If $u \in N_{\ssA_i}$ and $v \in N_{\ssA_{i'}}$ for some $i,i' \in [k_1]$ with $i \not= i'$ (resp.\ $u \in N_{\ssB_j}$ and $v \in N_{\ssB_{j'}}$ for some $j,j' \in [k_2]$ with $j \not= j'$) and $uv \in E$, then branch into two branches: in the first, open a new cluster in $\sB$ (resp.\ $\sA$) with $u$; in the second, open a new cluster in $\sB$ (resp.\ $\sA$) with $v$.
\end{brule}
\begin{proof}[Proof of correctness]
Let $(A,B)$ be any \valid\ 2-subcoloring of $G$ that fulfills $\C$. Suppose that $u \in N_{\ssA_i}$ and $v \in N_{\ssA_{i'}}$ for some $i,i' \in [k_1]$ with $i \not= i'$. Then $\ssA_i \subseteq A_i$ for some cluster $A_i$ of $G[A]$ and $\ssA_{i'} \subseteq A_{i'}$ for some cluster $A_{i'}$ of $G[A]$ by Proposition~\ref{prp:total-obs}, where $i\not=i'$. Since $uv\in E$, $u \not\in A_i$ or $v \not\in A_{i'}$. In fact, since $u$ (resp.\ $v$) is adjacent to all vertices of $\ssA_i$ (resp.\ $\ssA_{i'}$) by definition, at least one of $u,v$ is not in $A$. By definition, $u$ and $v$ are not adjacent to any vertex of $\sB$. Hence, at least one of $u,v$ must be in a different cluster of $G[B]$ than the vertices of $\sB$ by Proposition~\ref{prp:total-obs}. Therefore, $(A,B)$ fulfills the generated constraint in at least one of the branches. The argument in case $u \in N_{\ssB_j}$ and $v \in N_{\ssB_{j'}}$ for some $j,j' \in [k_2]$ with $j \not= j'$ follows symmetrically.
\smartqed\end{proof}
Below, let $I = \{ i \mid 1 \leq i \leq k_1, N_{\ssA_i} \not= \emptyset \}$ and let $J = \{ j \mid 1 \leq j \leq k_2, N_{\ssB_j} \not= \emptyset \}$.  If none of the above rules applies, then we have the following:

\begin{proposition} \label{prp:total-obs3}
If none of the above rules applies, then $G[\bigcup_{i\in I} N_{\ssA_i}]$ is a cluster graph whose clusters are $N_{\ssA_i}$, where $i\in I$,  and $G[\bigcup_{j \in J} N_{\ssB_j}]$ is a cluster graph whose clusters are $N_{\ssB_j}$, where $j \in J$.
\end{proposition}
\begin{proof}
If $N_{\ssA_i}$ is not a clique, then~\brref{br:total-2} applies. If a vertex of $N_{\ssA_i}$ is adjacent to a vertex of $N_{\ssA_{i'}}$ for $i \in I$ and $i' \in I \setminus \{i\}$, then~\brref{br:total-3} applies. The same holds mutatis mutandis with respect to $N_{\ssB_j}$.
\smartqed\end{proof}
The following branching rule is special in the sense that it will be applied only once in a root-leaf path of the corresponding search tree.

\begin{brule} \label{br:total-4}
Let $I = \{ i \mid 1 \leq i \leq k_1, N_{\ssA_i} \not= \emptyset \}$ and $J = \{ j \mid 1 \leq j \leq k_2, N_{\ssB_j} \not= \emptyset \}$. For each $I' \subseteq I$ and $J' \subseteq J$ such that $|I'| + |J'| \leq k-k_1-k_2$, branch into a branch where we:
\begin{itemize}
\item for each $i \in I'$, open a new cluster in $\sB$ with a fresh dummy vertex $s_i$ that is made adjacent to all vertices of $N_{\ssA_i}$;
\item for each $j \in J'$, open a new cluster in $\sA$ with a fresh dummy vertex $t_j$ that is made adjacent to all vertices of $N_{\ssB_j}$;
\item for each $i \in I \setminus I'$, add all vertices of $N_{\ssA_i}$ to $\ssA_i$; and
\item for each $i \in J \setminus J'$, add all vertices of $N_{\ssB_j}$ to $\ssB_j$.
\end{itemize}
\end{brule}
\begin{proof}[Proof of correctness]
  Recall Proposition~\ref{prp:total-obs3}. Let $(A,B)$ be any
  \valid\ 2-subcoloring of $G$ that fulfills
  $\C$. %
  Let~$I^*$ denote the set containing all integers~$i$ such that~$B$
  contains at least one vertex from~$N(\ssA_i)$. Similarly, let~$J^*$
  denote the set containing all integers~$j$ such that~$A$ contains at
  least one vertex from~$N(\ssB_j)$. Each integer in~$I^*\cup J^*$
  corresponds to a distinct cluster that does not intersect with any
  cluster of the constraint~$\C$. Thus, since~$(A,B)$
  is~\valid, we have $|I^*|+|J^*|\le k-k_1-k_2$. Consequently,
  there is a constraint~$\C'$ that was created in the branch
  where~$I'=I^*$ and~$J'=J^*$. Let~$G'$ denote the graph constructed
  for this branch, that is, the graph~$G$ plus the dummy
  vertices. Then, the bipartition~$(A':=A\cup \{s_i\mid i\in
  I'\},B':=B\cup \{t_j\mid j\in J'\})$ of~$G'$ is a solution
  fulfilling the constraint~$\C'$: For each~$i\in I\setminus
  I'$, $A'$ contains a cluster~$A_i=A^\C_i$.  Moreover, for each~$i\in
  I'$,~$B'$ contains a cluster~$B_i=N(\ssA_i)\cup \{s_i\}$. The
  clusters for each~$j\in J$ can be defined
  symmetrically. Since~$(A,B)$ is a 2-subcoloring of~$G$ and by the
  construction of the~$s_i$ and~$t_j$, there are no edges between
  different clusters in~$A$ or in~$B$. Hence,~$(A',B')$ is a
  2-subcoloring. Moreover, the bipartition~$(A',B')$ does not contain
  any further clusters and by the restriction on~$|I'|+|J'|$ it is
  thus~\valid.
\smartqed\end{proof}
This completes the description of the reduction and branching
rules. After exhaustive application of the rules, the problem of
finding a solution fulfilling the constraint associated with a node
can be solved by reduction to {\sc 2-CNF-Sat}.
\begin{definition}
  Say that a node in the search tree is \emph{exhausted} if none of the reduction and branching rules apply in that node.
\end{definition}

\begin{lemma}\label{lem:total-1}
  There is an algorithm that takes as input the
  graph~$G = (V, E)$ and a constraint
  $\C = (\ssA_{1},\ldots,\ssA_{k_{1}}, \ssB_{1},\ldots,\ssB_{k_{2}},
  \rC)$
  associated with an exhausted leaf node, and decides in $\Oh(n^2)$ time whether $G$ has
  a \valid\ 2-subcoloring that fulfills $\C$.
\end{lemma}
\begin{proof}
  We construct an instance of {\sc 2-CNF-Sat} (satisfiability of Boolean formulas in the conjunctive normal form in which each clause contains at most two literals) in $\Oh(n^2)$ time that can be satisfied if and only if $\C$ can be fulfilled. Since {\sc 2-CNF-Sat} is solvable in linear time (see \citet{papa} for example), the algorithm runs in $O(n^2)$ time.

  Observe that it suffices to place the vertices of $\rC$ into clusters. So let $v$ be a vertex in~$\rC$. Since no rules apply to $\C$ and in particular \brref{br:total-2} and \brref{br:total-3} do not apply, $v \in N_{\ssA_i\ssB_j}$ for some $i\in [k_1]$ and $j\in [k_2]$. Suppose that $(A,B)$ is a \valid\ 2-subcoloring that fulfills $\C$. It follows from Proposition~\ref{prp:total-obs} that if $v \in A$, then $v \in A_i$, and if $v \in B$, then $v \in B_j$. Therefore, what is left is to decide is whether the vertices in the sets $N_{\ssA_i\ssB_j}$ can be partitioned between the $A$ and $B$ in such a way that $v$ can only be placed in one of its two associated clusters, and such that the resulting bipartition is a 2-subcoloring of~$V$. We model this as an instance $\Phi$ of {\sc 2-CNF-Sat} as follows.

  For each vertex $v \in \rC$, we create a Boolean variable $x_v$. Assigning $x_v=1$ corresponds to adding $v$ to its associated cluster in $A$, and assigning $x_v=0$ corresponds to adding $v$ to its associated cluster in $B$. The Boolean formula $\Phi$ is constructed as follows. For every two vertices $v, v' \in \rC$, let (say) $v \in N_{\ssA_i\ssB_j}$ and $v' \in N_{\ssA_{i'}\ssB_{j'}}$. We add the clause $(\co{x_v} \vee \co{x_{v'}})$ to $\Phi$ if $v$ and $v'$ are adjacent but $i \neq i'$ or if $v$ and $v'$ are nonadjacent but $i=i'$; we add the clause $(x_v \vee x_{v'})$ to $\Phi$ if $v$ and $v'$ are adjacent but $j \neq j'$ or if $v$ and $v'$ are nonadjacent but $j=j'$. In both cases, the added clauses enforce that $v$ and $v'$ are assigned to different cluster-groups in any satisfying assignment of $\Phi$. This completes the construction of $\Phi$. It is not difficult to verify that $\C$ can be fulfilled if and only if $\Phi$ is satisfiable. Indeed, the correctness proof follows along similar lines as the correctness of \brrefm{br:total-2} and~\ref{br:total-3}.

  To carry out the construction of the 2-CNF formula in $\Oh(n^2)$ time, proceed as follows. First, construct the adjacency matrix of $G$ in $\Oh(n^2)$ time. Then, for each vertex $v \in \rC$ determine~$i$ and $j$ such that $v \in N_{\ssA_i\ssB_j}$. That is, compute an array that maps each vertex~$v$ to the tuple~$(i, j)$. This can be done in $\Oh(m)$ time by iterating over all neighbors of the clusters in $\sA$ and $\sB$, and setting the entries of the arrays corresponding to each neighbor accordingly. Finally, iterate over all pairs of vertices in $\rC$, determine whether they are adjacent in $\Oh(1)$ time, determine whether their $i$ and $j$-values differ in $\Oh(1)$~time using the computed array, and add the clause in $\Oh(1)$ time accordingly. Hence, the algorithm runs in $\Oh(n^2)$ time, as required.
\end{proof}

We now combine the reduction and branching rules with Lemma~\ref{lem:total-1} and prove the overall running-time bound.
\begin{reptheorem}{thm:main:sub2}
  \thmtxtmainsubtwo
\end{reptheorem}
\begin{proof}%
   As outlined in the beginning of this section, the algorithm works as follows. We start with the constraint $\C_r = (\rC_r = V)$. Then we apply the reduction and branching rules exhaustively, in the order they were presented. A leaf of the search tree is then a node of the search tree with associated constraint $\C$ such that either $\C$ is rejected or none of the rules applies to $\C$. The latter is an \emph{exhausted leaf}. Lemma~\ref{lem:total-1} shows that there is an $\Oh(n^2)$~time algorithm that decides if there is a \valid\ 2\nobreakdash-subcoloring of $G$ that fulfills the constraint associated with an exhausted leaf. If this algorithm accepts the constraint associated with at least one exhausted leaf, then we answer that $G$ has a \valid\ 2-subcoloring; otherwise, we answer that $G$ does not have a \valid\ 2-subcoloring.

The algorithm is correct, because by the correctness of the rules there is an exhausted leaf of the spanning tree for which the associated constraint can be fulfilled if and only if $G$ has a \valid\ 2-subcoloring. Lemma~\ref{lem:total-1} implies that an exhausted leaf will be accepted if and only if~$G$ has a \valid\ 2-subcoloring.

We now analyze the running time. We first claim that the algorithm never executes \brref{br:total-4} more than once in a root-leaf path of the search tree. By the description of the algorithm, none of the other branching rules are applicable if \brref{br:total-4} is applied. Furthermore, it is clear from the description of the branching rules that none of the branching rules are applicable after \brref{br:total-4} has been applied. Indeed, after \brref{br:total-4} has been applied, observe that every vertex of $\rC$ is in $N_{\ssA_i,\ssB_j}$ for some $i,j$ for the associated constraint $\C$ which contradicts the prerequisites of the other branching rules. As for the reduction rules, clearly, neither \rrref{rr:total-1} nor \rrref{rr:total-2} can become applicable after applying \brref{br:total-4}. The claim follows, that is, \brref{br:total-4} is applied at most once in a root-leaf path of the search tree.

To see that the search tree has $\Oh(4^k)$ nodes, note that \brrefm{br:total-1}, \ref{br:total-2}, and~\ref{br:total-3} are two-way branches, and each of these rules opens exactly one new cluster in each branch. By \rrref{rr:total-1}, at most $k+1$ clusters can be opened, and thus these branching rules lead to $\Oh(2^{k})$ nodes of the search tree. After all these rules have been applied, possibly \brref{br:total-4} will be applied. Since \rrref{rr:total-1} cannot be applied, \brref{br:total-4} yields at most $2^k$ branches. As shown above, no further branching rules will be applied after \brref{br:total-4} has been applied. Hence, the search tree has $\Oh(4^k)$ leaves.

Finally, we analyze the running time along a root-leaf path in the search tree. We look at all rules individually and combine their running time bounds with Lemma~\ref{lem:total-1}. For \rrref{rr:total-1}, clearly, whenever we open a cluster, we can check in $\Oh(1)$~time whether it becomes applicable and carry out the reduction rule accordingly. Hence, the time needed for \rrref{rr:total-1} is clearly bounded by $\Oh(n^2)$.

For \rrref{rr:total-2}, note that it only becomes applicable to a vertex~$v$ if either a neighbor is put into one of the clusters in $\sA$ or $\sB$, or if a new cluster is opened with a neighbor of~$v$. We claim that, whenever we open a cluster with a vertex~$u$ and whenever we put a vertex~$u$ into a cluster, we can decide in $\Oh(k\deg(u))$ time whether \rrref{rr:total-2} becomes applicable. The overall running time for checking and applying \rrref{rr:total-2} is then $\Oh(k(n + m))$ as each vertex is put into a cluster at most once. To decide whether \rrref{rr:total-2} becomes applicable when adding a vertex~$u$ to a cluster (or opening a cluster with~$u$), we check for each neighbor~$w$ of~$u$, whether it becomes adjacent to two clusters of $\sA$ or of $\sB$ or adjacent to a proper subset of the vertices of a cluster. This can be checked in $\Oh(k)$ time by maintaining, throughout the algorithm, the sizes of each cluster and, for each vertex~$x \in \rC$, two auxiliary arrays that keep track of which clusters~$x$ is adjacent to, and how many neighbors $x$ has in each cluster. The overall update cost for these auxiliary arrays is $\Oh(n + m)$ because whenever we put a vertex into a cluster we can update the auxiliary arrays for all neighbors. Hence, overall we need $\Oh(k(n + m))$ time to check for applicability and for applying \rrref{rr:total-2}, along a root-leaf path in the search tree.

\brref{br:total-1} can clearly be applied in $\Oh(n + m)$ time and is applied at most~$k$ times, yielding $\Oh(k(n + m))$ time along a root-leaf path in the search tree.

Next, in the course of the algorithm we need to compute (and possibly recompute several times) the partition of $\rC$ into the sets $N_{\ssA_i}$, $N_{\ssB_j}$ and $N_{\ssA_i\ssB_j}$. Using the information about the adjacency of vertices and clusters from the auxiliary arrays above, the partition can be computed in $\Oh(k(n + m))$ time. We recompute this partition whenever we have exhaustively applied \rrref{rr:total-1}, \rrref{rr:total-2} and \brref{br:total-1} and before we apply any of \brref{br:total-2}, \brref{br:total-3} or \brref{br:total-4}. Since we showed above that the branching rules are applied at most $2k + 1$ times, the time needed for computing the sets $N_{\ssA_i}$, $N_{\ssB_j}$ and $N_{\ssA_i\ssB_j}$ is $\Oh(k^2(n + m))$ along a root-leaf path in the search tree.

It is not hard to see that each of \brref{br:total-2}, \brref{br:total-3}, and \brref{br:total-4} can be applied in $\Oh(n + m)$ time. As was shown above, they are applied at most~$k + 1$ times, yielding an $\Oh(k(n + m))$ time needed for these rules along a root-leaf path in the search tree. To conclude, we need overall $\Oh(k^2(n + m))$ time to apply the reduction and branching rules along a root-leaf path in the search tree. Combining this with Lemma~\ref{lem:total-1} and the fact that the search tree has $\Oh(4^k)$~leaves, we obtain an overall running time of $\Oh(4^k \cdot (n^2 + k^2(n+ m))) = \Oh(4^k \cdot k^2 \cdot n^2)$.
\smartqed\end{proof}

\section{Hardness results}
\label{sec:hardness}
In this section, we present hardness results showing that significant improvements over the algorithms presented in the previous sections for \textsc{Monopolar Recognition}, \textsc{2-Subcoloring}, and the $(\Pi_{A},\Pi_{B})$-Recognition problem in general, are unlikely. These hardness results are proved under the assumption P $\neq$ NP, or the stronger assumption that the Exponential Time Hypothesis~(ETH) does not fail. We start with the following propositions showing that the existence of subexponential-time fixed-parameter algorithms for certain \textsc{$(\Pi_{A},\Pi_{B})$\nobreakdash-Recognition} problems is unlikely:

\begin{proposition}\label{prp:eth}
Let~$\Pi_A$ and~$\Pi_B$ be hereditary graph properties that can be characterized by a set of connected forbidden induced subgraphs. Moreover, assume that~$\Pi_A$ is not the set of all edgeless graphs. Then, \textsc{$(\Pi_A,\Pi_B)$-Recognition} cannot be solved in~$2^{o(n+m)}$ time, unless the ETH fails.
\end{proposition}
\begin{proof}
For every~$\Pi_A$ and~$\Pi_B$ fulfilling the conditions in the proposition,~\citet{Far04} presents a polynomial-time reduction from a variant of~\textsc{$p$-in-$r$ SAT}.
Herein, we are given a boolean formula~$\Phi$ with clauses of size~$r$ containing only positive literals, and the question is whether there is a truth assignment that sets exactly~$p$ variables in each clause to true. Given a formula~$\Phi$, the reduction of~\citet{Far04} constructs an instance~$G=(V,E)$ of~\textsc{$(\Pi_A,\Pi_B)$-Recognition} where~$n+m=\Oh(|\Phi|)$ (for each clause, the construction adds a gadget of constant size). The result now follows from the fact that \textsc{$p$-in-$r$ SAT} cannot be solved in~$2^{o(|\Phi|)}$ time, unless the ETH %
 fails~\cite{JLGZ13}.
\smartqed\end{proof}

The lower bounds on the running time for \textsc{Monopolar Recognition} and  \textsc{2-Subcoloring} now follow from the above proposition:

\begin{proposition}\label{prp:main:subexp}
\textsc{Monopolar Recognition} parameterized by the number $k$ of clusters in $G[A]$ and \textsc{2-Subcoloring} parameterized by the total number $k$ of clusters in $G[A]$ and $G[B]$ cannot be solved in $2^{o(k)}n^{\Oh(1)}$ time, unless the ETH fails.
\end{proposition}
\begin{proof}
\textsc{Monopolar Recognition} is the case of \textsc{$(\Pi_A,\Pi_B)$-Recognition} where $\Pi_A$ defines the class of $P_3$-free graphs and $\Pi_B$ defines the class of edgeless graphs. \textsc{2-Subcoloring} is the case of \textsc{$(\Pi_A,\Pi_B)$-Recognition} where $\Pi_A$~and~$\Pi_B$ define the class of $P_3$-free graphs. Hence, for both problems, $\Pi_A$~and~$\Pi_B$ satisfy the conditions in Proposition~\ref{prp:eth}. It now remains to observe that~$n$ is an upper bound on~$k$ in both cases.
\smartqed\end{proof}

In Theorems~\ref{thm:main:mono} and~\ref{thm:main:sub} we give fixed-parameter algorithms for two \textsc{$(\Pi_A,\Pi_B)$-Recognition} problems, in both of which $\Pi_A$ defines the set of all cluster graphs, parameterized by the number of clusters in $G[A]$. Hence, one might hope for a generic fixed-parameter algorithm for such problems, irrespective of~$\Pi_B$. However, polar graphs stand in our way. A graph~$G=(V,E)$ has a \emph{polar partition} if~$V$ can be partitioned into sets~$A$ and~$B$ such that~$G[A]$ is a cluster graph and~$G[B]$ is the complement of a cluster graph (a \emph{co-cluster graph})~\cite{TC85}. We have the following proposition:

\begin{proposition}\label{prp:main:polar}
It is NP-hard to decide whether~$G$ has a polar partition $(A,B)$ such that $G[A]$ is a cluster graph with one cluster or $G[B]$ is a co-cluster graph with one co-cluster.
\end{proposition}
\begin{proof}
Observe that a polar partition $(A,B)$ of a graph $G=(V,E)$ where $G[B]$ is a co-cluster graph with a single co-cluster implies that $G[A]$ is a cluster graph and $G[B]$~is edgeless. Hence, $G$ is monopolar. Since \textsc{Monopolar Partition} is NP-hard~\cite{Far04}, it follows immediately that it is NP-hard to decide whether $G$ has a polar partition $(A,B)$ such that $G[B]$ is a co-cluster graph with a single co-cluster. Now observe that the complement of a polar graph is again a polar graph. Hence, a straightforward NP-hardness reduction by taking the complement reveals that it NP-hard to decide whether $G$ has a polar partition $(A,B)$ such that $G[A]$ is a cluster graph with a single cluster.
\smartqed\end{proof}

\section{Conclusion}
\label{sec:conclusion}
In this paper, we developed the inductive recognition technique for designing algorithms for \textsc{$(\Pi_{A},\Pi_{B})$-Recognition} problems on graphs. Among other applications, we showed how this technique can be employed to design FPT algorithms for two graph partitioning problems: {\sc Monopolar Recognition} parameterized by the number of cliques (in the cluster graph part), and {\sc 2-Subcoloring}  parameterized by the smaller number of cliques between the two parts. These results generalize the well-known linear-time algorithm for recognizing split graphs and the polynomial-time algorithm for recognizing unipolar graphs, respectively. We believe that inductive recognition can be of general use for proving the fixed-parameter tractability of recognition problems that may not be amenable to other standard approaches in parameterized algorithmics. We also explored the boundaries of tractability for \textsc{$(\Pi_{A},\Pi_{B})$-Recognition} problems.

There are several open questions that ensue from our work. A natural concrete question is whether we can improve the running time of the \FPT\ algorithm for {\sc 2-Subcoloring} with respect to~$k$, the number of clusters $G[A]$. In particular, can we solve the problem in time  $2^{\Oh(k)}n^{\Oh(1)}$, or in time $f(k) n^2$? Note that the latter running time would match the quadratic running time for recognizing unipolar graphs, corresponding to the parameter value $k=1$. It is also interesting to investigate further if we can obtain meta \FPT-results for {\sc $(\Pi_A, \Pi_B)$-Recognition} based on certain graph properties $\Pi_A$ and $\Pi_B$, similar to the results obtained for mutually $d$-exclusive graph properties. Moreover, investigating the existence of polynomial kernels for {\sc Monopolar Recognition} and {\sc 2-Subcoloring} with respect to the considered parameters, and for {\sc $(\Pi_A, \Pi_B)$-Recognition} in general based on $\Pi_A$ and $\Pi_B$ and with respect to proper parameterizations, is certainly worth investigating.

One could also consider vertex-partition problems that are defined using more than two graph properties, that is, {\sc $(\Pi_A, \Pi_B, \Pi_C, \ldots)$-Recognition}. Observe that this problem is equivalent to {\sc $(\Pi_{A'}, \Pi_{B'})$-Recognition}, where $\Pi_{A'} = \Pi_{A}$ and $\Pi_{B'}$ is the graph property ``the vertices can be partitioned into sets that induce graphs with property $\Pi_B$, $\Pi_C$, \ldots, respectively''. For example, a graph is $d$-subcolorable if its vertices can be partitioned into $d$ sets that each induce a cluster graph, or alternatively, if its vertices can be partitioned into two sets, one of which is a cluster graph and the other of which is $(d-1)$-subcolorable. Using the aforementioned equivalence, the NP-hardness result of Farrugia~\cite{Far04} carries over to {\sc $(\Pi_A, \Pi_B, \Pi_C, \ldots)$-Recognition} whenever each graph property is additive and hereditary.

At first sight, considering the generalization of {\sc $2$-Subcoloring} to {\sc $d$-Subcoloring} seems promising. By the above discussion, {\sc $d$-Subcoloring} is NP-hard for any fixed $d \geq 2$ (this was also proved  explicitly by Achlioptas~\cite{Ach97}). A stronger result, however, follows from the observation that the straightforward reduction from {\sc $d$-Clique Cover} (or {\sc $d$-Coloring}) implies that {\sc $d$-Subcoloring} is NP-hard for any fixed $d \geq 3$, even when the total number of clusters is at most~$d$. This result implies that, unless \P${}={}$\NP, {\sc $d$-Subcoloring} has no FPT or XP algorithm for any fixed $d \geq 3$ when parameterized by the total number of clusters. Hence, unless \P${}={}$\NP, Theorems~\ref{thm:main:sub} and~\ref{thm:main:sub2} cannot be extended to {\sc $d$-Subcoloring} for any fixed $d \geq 3$.

To extend the results for {\sc Monopolar Recognition} and \textsc{2-Subcoloring}, one could thus instead consider the following two different problems, both of which are NP-hard per the above discussion:
\begin{itemize}
\item Given a graph $G = (V,E)$ and $k \in \mathbb{N}$, can $V$ be partitioned into a cluster graph with at most $k$ clusters and two edgeless graphs?

For $k=1$, this problem asks whether there is an independent set in $G$ whose removal leaves a split graph, and it can be solved in polynomial time~\cite{Brandstadt98}. Moreover, if $\Pi_A$ is the property of being a cluster graph with at most $k$ clusters and $\Pi_B$ is the property of being a bipartite graph, then this problem is a {\sc $(\Pi_{A}, \Pi_{B})$-Recognition} problem. Since properties~$\Pi_A$ and~$\Pi_B$ are mutually $(2k+1)$-exclusive, Theorem~\ref{thm:exclusive-rec} implies an XP algorithm for this problem when parameterized by $k$.

\item Given a graph $G = (V,E)$ and $k \in \mathbb{N}$, can $V$ be partitioned into two cluster graphs with at most $k$ clusters in total and one edgeless graph?

For $k=1$, this problem asks whether $G$ is a split graph, and it can be solved in polynomial time~\cite{HS81}. For $k \leq 2$, this problem asks whether there is an independent set in $G$ whose removal leaves a co-bipartite graph, and it can be solved in polynomial time as well~\cite{Brandstadt98}. Moreover, if $\Pi_A$ is the property of being $2$-subcolorable with at most $k$ clusters in total, and $\Pi_B$ is the property of being an edgeless graph, then this problem is a {\sc $(\Pi_{A}, \Pi_{B})$-Recognition} problem. Since properties~$\Pi_A$ and~$\Pi_B$ are mutually $(k+1)$-exclusive, Theorem~\ref{thm:exclusive-rec} implies an XP algorithm for this problem when parameterized by $k$.
\end{itemize}
Since both problems have XP algorithms when parameterized by~$k$, it is interesting to investigate whether any of them has an FPT algorithm. We leave these as open questions for future work.

\bibliographystyle{abbrvnat}
\bibliography{vertex-partition}

\end{document}